\documentclass[11pt]{article}
\usepackage{amsfonts,amsthm,amsmath,amssymb,latexsym,comment,enumerate}

\usepackage{comment}
\usepackage{makeidx}  % allows for indexgeneration
\usepackage[T1]{fontenc}
\usepackage[latin1]{inputenc}
\usepackage{epic}
\usepackage{eepic}
\usepackage{alltt}
\usepackage{xspace}
\usepackage{fancyhdr}
\usepackage{fullpage}
\usepackage{graphicx}
 \ifx\pdftexversion\undefined
  \usepackage[colorlinks,linkcolor=black,filecolor=black,
  citecolor=black,urlcolor=black,pdfstartview=FitH]{hyperref}
\else
  \usepackage[colorlinks,linkcolor=blue,filecolor=blue,
  citecolor=blue,urlcolor=blue,pdfstartview=FitH]{hyperref}
\fi

\def\beginsmall#1{\vspace{-\parskip}\begin{#1}\itemsep-\parskip}
\def\endsmall#1{\end{#1}\vspace{-\parskip}}

\newcommand{\ddu}{{\sffamily\bfseries\color{red}\mbox{**} update \mbox{**}}}
\newcommand{\hide}[1]{#1}
\renewcommand{\hide}[1]{}
\newcommand{\nhide}[1]{#1}

\def\squarebox#1{\hbox to #1{\hfill\vbox to #1{\vfill}}}

\newcommand{\tb}{\makebox[0.4cm]{}}
\newcommand{\due}{\makebox[0.8cm]{}}
\newcommand{\tre}{\makebox[1.2cm]{}}

\newcommand{\ignore}[1]{}

\newcommand{\ra}{\rangle}
\newcommand{\la}{\langle}

\newcommand\latest{\mbox{\it latest\_accept}}

\newcommand\broadcast{\mbox{\slshape i\_values}}
\newcommand\broadcaster{\broadcast[G,*]}

\newcommand\TAU{\tau^{\mbox{\sc{\tiny G}}}}

\newcommand\tskew{t^{\mbox{\sc\tiny G}}_{\mbox{\sc\tiny skew}}}

\newcommand\ByzAgreementl{\mbox{\large\sc ss-Byz-Agree}\xspace}
\newcommand\ByzAgreement{\mbox{\sc ss-Byz-Agree}\xspace}

\newcommand\initiator{\mbox{\sc Initiator-Accept}\xspace}

\newcommand\broadcastpl{\mbox{\large\sc msgd-broadcast}\xspace}
\newcommand\broadcastp{\mbox{\sc msgd-broadcast}\xspace}
\newcommand\iacpt{\mbox{{\sc I}-accept}\xspace}
\newcommand\iacpts{\mbox{{\sc I}-accepts}\xspace}
\newcommand\addrho{}
\newcommand\BYZdur{\dagr}
\newcommand\dz{\Delta_0}
\newcommand\dnet{\Delta_{\mbox{\footnotesize net}}}
\newcommand\dagr{\Delta_{\mbox{\footnotesize agr}}}
\newcommand\drmv{\Delta_{\mbox{\footnotesize rmv}}}
\newcommand\dstb{\Delta_{\mbox{\footnotesize stb}}}

\newcommand\dv{\Delta_{\mbox{\footnotesize v}}}
\newcommand\dnode{\Delta_{\mbox{\footnotesize node}}}
\newcommand\dreset{\Delta_{\mbox{\footnotesize reset}}}

\newcommand\support{(support,G,m)}
\newcommand\ready{(ready,G,m)}
\newcommand\approve{(approve,G,m)}

\newcommand\readygm{ready_{\!_{\mbox{\,\tiny G,m}}}}
\newcommand\lastgm{last(G,m)}

\newcommand{\namedref}[2]{\hyperref[#2]{#1~\ref*{#2}}}

\newcommand{\sectionref}[1]{\namedref{Section}{#1}}

\newcommand{\figureref}[1]{\namedref{Figure}{#1}}
\newcommand{\lemmaref}[1]{\namedref{Lemma}{#1}}
\newcommand{\corollaryref}[1]{\namedref{Corollary}{#1}}
\newcommand{\claimref}[1]{\namedref{Claim}{#1}}

\newtheorem{theorem}{Theorem}
\newtheorem{lemma}{Lemma}

\newtheorem{definition}{Definition}

\newtheorem{claim}{Claim}
\newtheorem{corollary}{Corollary}

\def\squarebox#1{\hbox to #1{\hfill\vbox to #1{\vfill}}}
\newcommand{\ie}{\emph{i.e.,\ }}
\begin{document}
%\onecolumn

%\conferenceinfo{PODC'06,} {July 22-26, 2006, Denver, Colorado, USA.}
%\CopyrightYear{2006}
%\crdata{1-59593-384-0/06/0007}

%\pagestyle{fancy}
%\dmyydate
%\rfoot{\today;~\currenttime}
%\lfoot{working draft}

%\pagestyle{headings}  % switches on printing of running heads

%\mainmatter
\title{Self-stabilizing Byzantine Agreement}
%\titlerunning{Self-stabilizing Byzantine Agreement}

%\numberofauthors{2}

%\author{
%\alignauthor Ariel Daliot
%\\\affaddr{School of Engineering and Computer Science}\\\affaddr{ The
%Hebrew University, Jerusalem, Israel} \\\email{
%adaliot@cs.huji.ac.il} \alignauthor Danny Dolev\titlenote{Part of
%the research was done while visiting Cornell University. Supported in  part by ISF.}
%\\\affaddr{School of Engineering and Computer Science}\\\affaddr
%{The Hebrew University, Jerusalem, Israel}
%\\\email{ dolev@cs.huji.ac.il} }

\author{
Ariel Daliot and Danny Dolev\footnote{Supported in  part by ISF.}
\\School of Engineering and Computer Science\\
The
Hebrew University, Jerusalem, Israel \\
Email:\ adaliot@gmail.com, danny.dolev@mail.huji.ac.il  }
%\date{}

\maketitle
\begin{abstract}
Byzantine agreement algorithms typically assume implicit initial
state consistency and synchronization among the correct nodes and
then operate in coordinated rounds of information exchange to reach
agreement based on the input values. The implicit initial
assumptions enable correct nodes to infer about the progression of
the algorithm at other nodes from their local state. This paper
considers a more severe fault model than permanent Byzantine
failures, one in which the system can in addition be subject to
severe transient failures that can temporarily throw the system out
of its assumption boundaries. When the system eventually returns to
behave according to the presumed assumptions it may be in an
arbitrary state in which any synchronization among the nodes might
be lost, and each node may be at an arbitrary state. We present a
self-stabilizing Byzantine agreement algorithm that reaches
agreement among the correct nodes in an optimal ration of faulty to correct, by using only the
assumption of eventually bounded message transmission delay. In the process of
solving the problem, two additional important and challenging
building blocks were developed: a unique self-stabilizing protocol
for assigning consistent relative times to protocol initialization
and a Reliable Broadcast primitive that progresses at the speed of
actual message delivery time.
\end{abstract}

\vspace{1mm} \noindent{\bf Categories and Subject Descriptors:}
C.2.4 {[Distributed Systems]}: {Distributed applications};

\vspace{1mm} \noindent {\bf General Terms:} Algorithms, Reliability,
Theory.

\vspace{1mm} \noindent {\bf Keywords:} Byzantine Agreement,
Self-Stabilization, Byzantine Faults, Pulse Synchronization,
Transient Failures,  Reliable Broadcast.

%\thispagestyle{empty}
%\newpage
%\setcounter{page}{1}
%\newpage
%\pagenumbering{arabic}
\section{Introduction}
\label{sec:intro} The Byzantine agreement (Byzantine Generals)
problem was first introduced by Pease, Shostak and Lamport
\cite{Agree80}. It is now considered as a fundamental problem in
fault-tolerant distributed computing. The task is to reach agreement
in a network of $n$ nodes in which up-to $f$ nodes may be faulty. A
distinguished node (\emph{the General} or \emph{the initiator})
broadcasts a value \emph{m}, following which all nodes exchange
messages until the non-faulty nodes agree upon the same value. If
the initiator is non-faulty then all non-faulty nodes are required
to agree on the same value that the initiator sent.

\ignore{On-going faults whose nature is not predictable or that
express complex behavior are most suitably addressed in the
Byzantine fault model. It is the preferred fault model in order to
seal off unexpected behavior within limitations on the number of
concurrent faults. With respect to the bounds on redundancy, the
Byzantine agreement problem has been shown to have no deterministic
solution if more than $n/3$ of the nodes are concurrently faulty
\cite{Agree82}.}

Standard deterministic Byzantine agreement algorithms operate in the
synchronous network model in which it is assumed that all correct
nodes initialize the agreement procedure (and any underlying
primitives) at about the same time. By assuming concurrent
initializations of the algorithm a synchronous rounds structure can
be enforced and used to infer on the progression of the algorithm
from the point of initialization. Moreover, there is always an
implicit assumption about the consistency of the initial states of
all correct nodes, or at least a quorum of them.

We consider a more severe fault-model in which in addition to the
permanent presence of Byzantine failures, the system can also be
subject to severe transient failures that can temporarily throw all
the nodes and the communication subsystem out of the assumption
boundaries. E.g. resulting in more than one third of the nodes being
Byzantine or messages of non-faulty nodes getting lost or altered.
This will render the whole system practically unworkable. Eventually
the system must experiences a tolerable level of permanent faults
for a sufficiently long period of time. Otherwise it would remain
unworkable forever. When the system eventually returns to behave
according to the presumed assumptions, each node may be in an
arbitrary state. It makes sense to require a system to resume
operation after such a major failure without the need for an outside
intervention to restart the whole system from scratch or to correct
it.

Classic Byzantine algorithms cannot guarantee to execute from an
arbitrary state, because they are not designed with
self-stabilization in mind. They typically make use of assumptions
on the initial state of the system such as assuming all clocks are
initially synchronized or that the initial states are initialized
consistently at all correct nodes (cf. from the very first
polynomial solution~\cite{PolyAgree82} through many others
like~\cite{FastAgree87}). Conversely, A \textit{self-stabilizing}
protocol converges to its goal from any state once the system
behaves well again, but is typically not resilient to the permanent
presence of faults.

In trying to combine both fault models, Byzantine failures present a
special challenge for designing self-stabilizing distributed
algorithms due to the ``ambition'' of malicious nodes to incessantly
hamper stabilization. This difficulty may be indicated by the
remarkably few algorithms resilient to both fault models (see
\cite{ByzStabilizer} for a review). The few published
self-stabilizing Byzantine algorithms are typically complicated and
sometimes converge from an arbitrary initial state only after
exponential or super exponential time (\cite{DolWelSSBYZCS04}). Recently efficient solutions were presented for the strict synchronization model in which an outside entity provides repetitive synchronized timing events at all correct nodes at once (\cite{PulseSync-DISC07}).

\looseness -1 In our model correct nodes cannot assume a common
reference to time or even to any common anchor in time and they
cannot assume that any procedure or primitive initialize
concurrently. This is the result of the possible loss of
synchronization following transient faults that might corrupt any
agreement or coordination among the correct nodes and alter their
internal states. Thus synchronization must be restored from an
arbitrary state while facing on-going Byzantine failures. This is a
very tricky task considering that all current tools for containing
Byzantine failures, such as \cite{FSQUAD85,FastAgree87}, assume that
synchronization already exists and are thus preempted for use. Our
protocol achieves self-stabilizing Byzantine agreement without the
assumption of any existing synchrony besides bounded message
delivery. In \cite{FTSS97} it is proven to be impossible to combine
self-stabilization with even crash faults without the assumption of
bounded message delivery.

Note that the problem is not relaxed even in the case of a one-shot
agreement, i.e. in case that it is known that the General will
initiate agreement only once throughout the life of the system. Even
if the General is correct and even if agreement is initiated after
the system has returned to its coherent behavior following transient
failures, then the correct nodes might hold corrupted variable
values that might prevent the possibility to reach agreement. The
nodes have no knowledge as to when the system returns to coherent
behavior or when the General will initiate agreement and thus cannot
target to reset their memory exactly at this critical time period.
Recurrent agreement initialization by the General allows for
recurrent reset of memory with the assumption that eventually all
correct nodes reset their memory in a coherent state of the system
and before the General initializes agreement. This introduces the
problem of how nodes can know when to reset their memory in case of
many ongoing concurrent invocations of the algorithm, such as in the
case of a faulty General disseminating several values all the time.
In such a case correct nodes might hold different sets of messages
that were sent by other correct nodes as they might reset their
memory at different times.

In our protocol, once the system complies with the theoretically
required bound of $3f<n$ permanent Byzantine faulty nodes in a
network of $n$ nodes and messages are delivered within bounded time,
following a period of transient failures, then regardless of the
state of the system, the goal of Byzantine agreement is satisfied
within $O(f')$ communication rounds (where $f^\prime \le f$ is the
actual number of concurrent faults). The protocol can be executed in
a one-shot mode by a single General or by recurrent agreement
initializations and by different Generals. It tolerates transient
failures and permanent Byzantine faults and makes no assumption on
any initial synchronized activity among the nodes (such as having a
common reference to time or a common event for triggering
initialization).

For ease of following the arguments and proofs, the structure and
logic of our $\ByzAgreement$ procedure is modeled on that of
\cite{FastAgree87}. The rounds in that protocol progress following
elapsed time. Each round spans a constant predefined time interval.
Our protocol, besides being self-stabilizing, has the additional
advantage of having a message-driven rounds structure and not
time-driven rounds structure. Thus the actual time for terminating
the protocol depends on the actual communication network speed and
not on the worst possible bound on message delivery time.

It is important to note that we have previously presented a
distributed  self-stabilizing Byzantine pulse synchronization
procedure in~\cite{bio-pulse-synch}. It aims at delivering a common
anchor in time to all correct nodes within a short time following
transient failures and with the permanent presence of Byzantine
nodes. We have also previously presented a protocol for making any
Byzantine algorithm  be self-stabilizing \cite{DDPBYZ-CS03},
assuming the existence of synchronized pulses. Byzantine agreement
can easily be achieved using a pulse synchronization procedure: the
pulse invocation can serve as the initialization event for round
zero of the agreement protocol. Thus any existing Byzantine
agreement protocol may be used, on top of the pulse synchronization
procedure, to attain self-stabilizing Byzantine agreement. The
current paper achieves Byzantine agreement without assuming
synchronized pulses. Moreover, we show in \cite{DDSSPS} that
synchronized pulses can actually be produced more efficiently atop
the protocol in the current paper. This pulse synchronization
procedure can in turn be used as the pulse synchronization mechanism
for making any Byzantine algorithm self-stabilize, in a more
efficient way and in a more general model than by using the pulse
synchronization procedure in \cite{bio-pulse-synch}.

An early version of the results covered in the current paper appeared in~\cite{DDSSBA-PODC06}. The current paper provides elaborated proofs and correct some mistakes that appear in the early version.

In~\cite{Widder-BootCS03} it is shown how to initialize Byzantine
clock synchronization without assuming a common initialization
phase. It can eventually also execute synchronized Byzantine
agreement by using the synchronized clocks. The solution is not
self-stabilizing as nodes are booted and thus do not initialize with
arbitrary values in the memory. \ignore{It has on the other hand a
constant convergence time with respect to the required rounds of
communication, whereas our current solution has a dependency on
$f',$ which is due to the self-stabilization requirement.}

In~\cite{EventualSynch05} consensus is reached assuming eventual
synchrony. Following an unstable period with unbounded failures and
message delays, eventually no node fails and messages are delivered
within bounded, say $d,$ time. At this point there is no synchrony
among the correct nodes and they might hold copies of obsolete
messages. This is seemingly similar to our model but the solution is
not truly self-stabilizing since the nodes do not initialize with
arbitrary values. Furthermore, the solution only tolerates stopping
failures and no new nodes fail subsequent to stabilization.
Consensus is reached within $O(d).$  That paper also argues that in
their model, although with Byzantine failures, consensus cannot be
reached within less than $O(f')\cdot d$ time, which is essentially
identical to our time complexity. Our solution operates in a more
severe fault model and thus converges in linear  time.

%\vspace{-2mm}
\section{Model and Problem Definition}
\label{sec:model}  The environment is a network of $n$ nodes that
communicate by exchanging messages.
%The nodes regularly invoke ``pulses'', ideally
%every $\Cycle$ real-time units. The invocation of the pulse is
%expressed by sending a message to all the nodes.,
We assume that the message passing medium allows for an
authenticated identity of the senders.
The communication network
does not guarantee any order on messages among different nodes,
%dd added
though, when the network is functioning correctly, any message sent will eventually be delivered.
Individual nodes have no access to a central clock and there is no
external pulse system. The hardware clock rate (referred to as the
{\em physical timers}) at each non-faulty node has a bounded drift, $\rho,$
from real-time rate. Ensuant to transient failures there can be an
unbounded number of concurrent faulty nodes, the turnover rate
between faulty and non-faulty nodes can be arbitrarily large and the
communication network may behave arbitrarily.

%\newpage
%\vspace{-2mm}
\begin{definition} A node is {\bf non-faulty} at times
that it complies with the following:
\begin{enumerate}
\vspace{-0.5em}
\item \emph{(Bounded Drift)} Obeys a global constant
$0<\rho<1$ (typically $\rho \approx 10^{-6}$), such that for every
real-time interval $[u,v]:$ $$(1-\rho)(v-u)  \le \mbox{`physical timer'}(v) - \mbox{ `physical timer'}(u) \le (1+\rho)
(v-u).$$

\vspace{-2mm}
\item \emph{(Obedience)} Operates according to the instructed protocol.

\vspace{-2mm}
\item \emph{(Bounded Processing Time)} Processes any message of the instructed
protocol within $\pi$ real-time units of arrival time.\footnote{We
assume that the bounds include also the overhead of the operating
system in sending and processing of messages.}

\end{enumerate}
\end{definition}

A node is considered {\bf faulty} if it violates any of the above
conditions. A faulty node may recover from its Byzantine behavior
once it resumes obeying the conditions of a non-faulty node. In
order to keep the definitions consistent, the ``correction'' is not
immediate but rather takes a certain amount of time during which the
non-faulty node is still not counted as a correct node, although it
supposedly behaves ``correctly''.\footnote{For example, a node may
recover with arbitrary variables, which may violate the validity
condition if considered correct immediately.} We later specify the
time-length of continuous non-faulty behavior required of a
recovering node to be considered {\bf correct}.

\begin{definition}The communication network is {\bf non-faulty}
at periods that it complies with the
following:
\begin{enumerate}
\item Any message
arrives at its destination node within $\delta$ real-time units;

\item The sender's identity and the content of any message being received is not
tampered.
\end{enumerate}
\label{def:net_nonfaulty}

\end{definition}

%\vspace{-2mm}
Thus, our communication network model is a
``bounded-delay'' communication network. We do not assume the
existence of a broadcast medium. We assume that the network cannot
store old messages for arbitrary long time or lose any more
messages, once it becomes non-faulty.\footnote{A non-faulty network
might fail to deliver messages within the bound but will be masked
as a fault and accounted for in the $f$ faults. Essentially, we
assume that messages among correct nodes are delivered within the
time bounds.}

%up to here is taken from the pulse paper

We use the notation $d\equiv( \delta + \pi)\times(1+\rho).$ Thus, when the
communication network is non-faulty, $d$ is the upper bound on the
elapsed time from the sending of a message by a non-faulty node
until it is received and processed by every non-faulty
node, as measured by the local clock at any non-faulty node.\footnote{Nodes that were not faulty when the message was
sent.}

Note that $n,$ $f$ and $d$ are fixed constants and thus non-faulty
nodes do not initialize with arbitrary values of these constants.

A recovering node should be considered correct only once it has been
continuously non-faulty for enough time to enable it to have deleted
old or spurious messages and to have exchanged information with the
other nodes.

\begin{definition} The communication network is {\bf correct}
following $\dnet$ real-time of continuous non-faulty
behavior.\footnote{We assume $\dnet\ge d.$}
\end{definition}

\begin{definition} A node is {\bf correct} following $\dnode$  real-time of continuous
non-faulty behavior during a period that the communication network
is correct.\footnote{ $\dnode$ is defined in the next section}
\end{definition}

\begin{definition}\label{def:system-coherence} \emph{(System Coherence)} The system is said to
be {\bf coherent} at times that it complies with the following:
%\\
%
\beginsmall{itemize}
\item 
\emph{(Quorum)} There are at least $n-f$ correct nodes,\footnote{The condition can be
replaced by $(n+f)/2$ correct nodes with some modifications to the
structure of the protocol.} where $f$ is the upper bound on the
number of potentially non-correct nodes at steady state.
%the following is implied by having correct nodes
%\item \emph{(Network Correctness)} The communication network is correct.
\endsmall{itemize}
\end{definition}

Hence, when the system is not coherent,  there can be an unbounded
number of concurrent faulty nodes; the turnover rate between the
faulty and non-faulty nodes can be arbitrarily large and the
communication network may deliver messages with unbounded delays, if
at all. The system is considered coherent, once the communication
network and a sufficient fraction of the nodes have been non-faulty
for a sufficiently long time period for the pre-conditions for
convergence of the protocol to hold. The assumption in this paper,
as underlies any other self-stabilizing algorithm, is that the
system eventually becomes coherent.

\begin{definition}\label{def:system-stable} \emph{(System Convergence)} The system is said to
be {\bf stable} at times that it complies with the following:
%\\
%
\beginsmall{itemize}
\item 
\emph{(converging)} The system has been coherent for $\dstb$ time units;\footnote{We define $\dstb$ in the next section.}
\item
\emph{(stability)} The system remained coherent since that time.\endsmall{itemize}
\end{definition}

It is assumed that each node has a local timer that proceeds at the
rate of real-time.  The actual reading of the various timers may be
arbitrarily apart, but their relative rate is bounded in our model.
 To distinguish between a real-time value and a node's
local-time reading we use $t$ for the former and $\tau$ for the
latter. The function $rt(\tau_p)$ represents the real-time when the
timer of a non-faulty node $p$ reads (or read) $\tau_p$ at the current execution.

Observe that the local time at a node may wrap around, since we assume transient faults.  
The protocol and the primitives presented below require measuring only intervals of times.  It 
is assumed that the local time wrap around is larger than a constant factor of the maximal 
interval of time need to be measured. This way a node can uniquily measure any necessary 
intervals of time.

Since nodes measure only intervals  of time that span
several $d$, and $d$ itself includes a worst case drift factor, by definition, then  $d$ is an upper bound on the
time it takes to send and process messages  among correct nodes, measured by each local
timer, \ie including the drift factor.

%\newpage
\section{The \ByzAgreementl protocol}\label{sec:byzagree}

We consider the Byzantine agreement problem in which a
\emph{General} broadcasts a value and the correct nodes agree on the
value broadcasted.  In our model any node can be a General. An
instance of the protocol is executed per General, and a correct
General is expected to send one value at a time.\footnote{One can
expand the protocol to a number of concurrent invocations by using
an index to differentiate among the concurrent invocations.}
The
target is for the correct nodes to associate a local-time with the
protocol initiation by the General and to agree on a specific value
associated with that initiation, if they agree that such an
initiation actually took place. There is a bound on
how frequent a correct General
may initiate agreements, though Byzantine nodes
might try to trigger agreements on their values at an arbitrary rate.

The \ByzAgreement protocol is composed of the Agreement procedure (the main body of the protocol)
and two primitives: the  primitive \initiator and the \broadcastp one (as detailed later).
The General, $G$, initiates an agreement on a value $m$ by disseminating the message
$(Initiator, G, m)$ to all nodes. Upon receiving the General's
message, each node invokes the \ByzAgreement protocol, which in
turn invokes the primitive \initiator. Alternatively, if a
correct node did not receive the General's message but concludes that enough nodes have invoked the protocol
(or the primitive) it will participate by executing the
appropriate
parts of the primitive \initiator (but will not invoke it),
and following the completion of the primitive that node  may participate in the corresponding parts of the agreement procedure.
 
We will prove the following properties of the \ByzAgreement protocol. When the system is stable, if
all correct nodes invoke the protocol within a ``small'' time-window, as
will happen if the General is a correct node, then it is ensured
that the correct nodes agree on a value for the General. 
If the General is a correct node, the agreed value will be the value sent by the General. 
When not all
correct nodes happen to invoke the \ByzAgreement protocol within a
small time-window, as can happen if the General is faulty, then if
any correct node accepts a non-null value, all correct nodes will
accept and agree on that value.

\begin{figure*}[t] \center \fbox{\begin{minipage}{5.1in}
\footnotesize
\begin{alltt}
\setlength{\baselineskip}{3.5mm}
\noindent Protocol {\bf \ByzAgreement} on ($G,m$)
\mbox{\ }\hfill\textit{/* Executed at node $q.$ $\tau_q$ is the local-time at $q.$ */}\\
%dd changed order and added details to the comments below
%\mbox{\ }\hfill\textit{ /* Invoked at
%node $q$ upon arrival of a message (Initiator$,G,m$) from node $G$. */}\\
\mbox{\ }\hfill\textit{/* Block~Q is executed only when (and if) invoked. */}\\
\mbox{\ }\hfill\textit{/* The rest is executed following a setting of a value to $\TAU_q$. */}\\
%dd updated to:
\mbox{\ }\hfill\textit{/* At most one of blocks R through U is executed per such a setting of  $\TAU_q$. */ }\\

Q0.\ \ If $q=G$ then send (Initiator$,G,m$) to all .\mbox{\ }\hfill\textit{ /*
initiation of the primitive by the leader */}\\
\\

Q1.\ \  If received (Initiator$,G,m$) invoke {\bf \initiator}$(G,m)$.\\\mbox{\ }\hfill\textit{ /*
determines
$\TAU_q$ and a value $m'$ for node $G$ */}\\
\\
R1.\ \ {\bf if} {\bf \iacpt} $\la G,m',\TAU_q\ra$ {\bf and} $\tau_q-\TAU_q\le 4d$ {\bf then}\\
R2.\ \tre $value:=\la G,m'\ra;$\\
R3.\ \tre {\bf \broadcastp}$(q,value,1)$;\\
R4.\ \tre {\bf stop} and {\bf return} $\la value, \TAU_q\ra$. \\

S1.\ \tb {\bf if} by $\tau_q,$  $\tau_q\le
\TAU_q+(2r+1)\cdot\Phi,$ \\\due\due{\bf accepted} $r$
%dd updated
%distinct messages $(p_i,\la G,m'^\prime\ra,\tau_i,i),$ \\
%\ \tre\tb where $\forall i,j\;\;1\le i \le r,$ and $p_i\neq p_j\neq G$ {\bf then}\\
distinct messages $(p_i,\la G,m'^\prime\ra,i),$ $1\le i \le r$, \\
\ \tre\tb where $ \forall i,j$ $1\le i,j \le r$ and $p_i\neq p_j\neq  G,$ {\bf then}\\
S2.\ \tre  $value:=\la G,m'^\prime\ra;$ \\
S3.\ \tre {\bf \broadcastp}$(q,value,r+1)$;\\
S4.\ \tre {\bf stop} and {\bf return} $\la value, \TAU_q\ra$. \\

T1.\ \tb {\bf if} by $\tau_q,$ $\tau_q>
\TAU_q+(2r+1)\cdot\Phi$,
 $|broadcasters| < r-1$  {\bf then}\\
T2.\ \tre{\bf stop} and {\bf return} $\la \perp, \TAU_q\ra$. \\

U1.\ \tb {\bf if} $\tau_q> \TAU_q+(2f+1)\cdot\Phi$ {\bf then}\hfill\textit{}\\
U2.\ \tre{\bf stop} and {\bf return} $\la \perp, \TAU_q\ra$. \\
\\
{\bf cleanup:}\\
\tb -- Erase any value or message older than $(2f+1)\cdot\Phi+3d$
time units.\\
\tb -- $3d$ after returning a value reset  {\bf \initiator}, $\TAU_q,$  and {\bf \broadcastp}.
\end{alltt}
\normalsize
\end{minipage} }
\caption{The \ByzAgreement protocol}\label{alg:Byz-alg}
\end{figure*}

For ease of following the arguments and the logic of our
\ByzAgreement protocol, we chose to follow the building-block
structure of \cite{FastAgree87}. 
The
primitive \broadcastp,  presented in
\sectionref{sec:reliable-bdcst}, replaces
the broadcast
primitive that simulates authentication in~\cite{FastAgree87}.
The main differences between the
original synchronous broadcast primitive and \broadcastp are
two-folds: first, the latter executes rounds that are anchored at
some agreed event whose local-time is supplied to the primitive
through a parameter; second, the conditions to be satisfied at each
round at the latter need to be satisfied by some time span that is
a function of the round number and need not be executed only during
the round itself. This allows nodes to rush through the protocol in
the typical case when messages among correct nodes happen to be delivered faster than
the worse case round span.

The \ByzAgreement protocol  needs to take into consideration that
correct nodes may invoke the agreement procedure at arbitrary times
and with no knowledge as to when other correct nodes may have
invoked the procedure. A mechanism is thus needed to make all
correct nodes attain some common notion as to when the General may have sent a value, and what that value is. The differences of the real-time
representations of the different nodes' estimations should be
bounded. This mechanism is satisfied by the primitive \initiator,
presented in
\sectionref{sec:initiator}. The use of this initial step in the protocol provides the nodes with an initial potential value of the General, and as a result number of ``rounds'' necessary to reach agreement is two less than those of~\cite{FastAgree87}.

%:basic variable defs
We use the following notations in the description of the agreement
procedure
%dd added
and the related primitives: \vspace{-0.5em}
\begin{itemize}

\item Let $\Phi$ be the duration of time equal to
$(\TAU_{skew}+ 2d)$ local-time units on a correct node's timer,
where $\TAU_{skew}=6d$ in the context of this paper. Intuitively,
$\Phi$ is the duration of a ``phase'' on a correct node's timer.

\vspace{-2mm}\item $\dagr,$ the upper bound on the time it takes to run the agreement protocol, will be equal 
to $(2f+1)\cdot\Phi.$

\item $\dz=13d$, the minimal time between consecutive invocations of the protocol by the General, for  
different values.

\vspace{-2mm}\item $\drmv=(\dagr+\dz),$ the time after which old values are decayed.

\vspace{-2mm}\item $\dv=(15d+2\drmv)$,  the minimal time between two invocations of the protocol by the General, for the 
same value.

\vspace{-2mm}\item $\dnode=\dv+\dagr,$ the time it takes for a non-faulty node to be considered correct.

\vspace{-2mm}\item $\dreset=20d+4\drmv,$   the time during which the General sends nothing, when it notices a failure in 
agreeing on a value it sent.

\vspace{-2mm}\item $\dstb=2\dreset,$  stabilization time of the system.

\vspace{-2mm}\item $\perp$ denotes a null value.

\vspace{-2mm}\item In the primitive \initiator:

\begin{itemize}

\vspace{-3mm}\item An \emph{\iacpt}\footnote{An \emph{accept} is
issued within \broadcastp.} is issued on values sent by $G$.

\vspace{-1mm}\item $\TAU_q$ denotes the local-time estimate, at node
$q,$ as to when the General
%dd changed
%have sent a value that has been \iacpt
has sent the value associated with the \iacpt
%in \initiator
by node $q.$

\end{itemize}

\end{itemize}

In the context of this paper we assume that a correct General conform with the following criteria when sending its messages.\\

\noindent{\bf Sending Validity Criteria:} A non-faulty General $G$ sends (Initiator$,G,m$) provided that:
\begin{itemize}
\item[{[IG1]}] At least $\dz$ time passed from the sending of the previous initiation message by $G$.
\item[{[IG2]}] At least $\dv$ time passed from the sending  of previous initiation message with the same value $m$ by $G$.
\end{itemize}

Notice that both limitations can be circumvented by adding counters to 
concurrent agreement initiations.  The difference between the two cases has to do with the ability to converge from an arbitrary initial state.  If a node can send the same message again and again repeatedly, there is a way for the adversary to confuse of convergence protocol, as can be seen in the next section.

\begin{definition} We say:
\beginsmall{itemize}
\item A node $p$ {\bf decides} at time $\tau$ if it
stops at that local-time and returns $value \ne \perp.$

\item A node $p$ {\bf aborts} if it stops and returns
$\perp.$

\item A node $p$ {\bf returns} a value if it either aborts or
decides.
\endsmall{itemize}
\end{definition}

The \ByzAgreement protocol is presented (see
\figureref{alg:Byz-alg}) in a somewhat different style than the
original protocol in \cite{FastAgree87}. Each round has a
precondition associated with it: if the local timer value associated
with the initialization by the General is defined and the
precondition holds then the step is to be executed. It is assumed
that the primitives' instances invoked as a result of the
\ByzAgreement protocol are implicitly associated with the agreement
instance that invoked them. A node stops participating in the
procedures once it returns a value and it stopped participating in the invoked primitives $3d$ 
time units after that. We use the term participate to refer to a node that executes the 
protocol's (and primitives') steps.  The term invoke will refer to a node that also executes the 
first block of the protocol (Block~Q) or primitive (Block~K), as each correct node would do if 
the General is a correct one. A node accumulates messages associated with the protocol even before it invokes it or participates in it. Such messages are decayed if the node doesn't invoke or participate in the protocol, or being processed once it does.

The \ByzAgreement protocol satisfies the following typical
properties, provided that the system is stable:

\noindent{\bf Agreement:} If any connect node decides $(G,m)$, all correct nodes decide the same;

\noindent{\bf Validity:} If the General invokes \ByzAgreement then each correct node decides on the value sent by $G$;

\noindent{\bf Termination:} The protocol terminates in a finite
time.\\

Note that in light of our definitions, the Agreement property actually says that if the protocol returns a value~$\neq\perp$ at any correct node,  it returns the same value at all correct nodes.

The \ByzAgreement protocol also satisfies the following timing properties:\\

\noindent{\bf Timeliness:}
\vspace{-1mm}
\begin{enumerate}
\item (agreement)
If a correct node $q$ decides on
$(G,m)$ at $\tau_q$ then any correct node $q^\prime$ decides on
$(G,m)$ at some $\tau_{q^\prime}$ such that,
\begin{enumerate}
\item
$|rt(\tau_q)-rt(\tau_{q^\prime})|\le 3d,$ and if validity holds,
then $|rt(\tau_q)-rt(\tau_{q^\prime})|\le 2d;$
\item
$|rt(\tau^G_q)-rt(\tau^G_{q^\prime})|\le 6d;$
\item
$rt(\tau^G_q), rt(\tau^G_{q^\prime}) \in [t_1-2d,t_2],$ where
$[t_1,t_2]$ is the interval within which each correct node, $p$,  that obtained the $\tau^G_p$ appearing in (b) following the 
invocation of \ByzAgreement$(G,m)$, did so;
\item $rt(\tau^G_q)\le rt(\tau_q)$ and $rt(\tau_q)-rt(\tau^G_q)\le \BYZdur$.
\end{enumerate}

\item (validity)
If all correct nodes invoked the protocol in an interval
$[t_0,t_0+d],$ as a result of some value $m$ sent by a correct
General $G$ that conform with the Sending Validity Criteria, then for every
correct node $q$, the decision time $\tau_q$, satisfies $t_0-d \le
rt(\tau^G_q)\le rt(\tau_q)\le t_0+4d.$

\item (termination) The protocol terminates within $\dagr$ time units of
invocation, and within $\dagr+7d$ in case it was not invoked
explicitly.

\item (separation) Let $p$ and $q$ be two correct nodes that decided on  agreements regarding $G$,
then
% $t_2+5d<\bar t_1$ and $rt(\tau_q)+5d< \bar t_1<rt(\bar\tau_q),$
%where $t_2$ is the latest time at which a correct node invoked
%\ByzAgreement in the earlier agreement and $\bar t_1$ is the
%earliest \ByzAgreement invoked by a correct node in the later
%agreement. \dd{define how to associate decisions with invocations?}
%
\begin{enumerate}
 \item  
  for $m\neq m^\prime,$ 
    $|rt(\TAU_q)-rt(\TAU_p)|> 4d;$ 
  \item  
  for $m= m^\prime,$ either
    $|rt(\TAU_q)-rt(\TAU_p)|\leq 6d$ or $|rt(\TAU_q)-rt(\TAU_p)|>2\drmv-3d$
   \end{enumerate}
\end{enumerate}

Note that the bounds in the above property is with respect to $d$,
the bound on message transmission time among correct nodes and not
the worse case deviation represented by $\Phi$.

Observe that since there is no prior notion of the possibility that
a value may be sent, it might be that some nodes associate a $\perp$
with a faulty sending and others may not notice the sending at all.

The proof that the \ByzAgreement protocol meets its properties
appears in \sectionref{ssec:jbyzagree}.

\section{The primitive \initiator}\label{sec:initiator}

In a typical agreement protocol a General that wants to send
some value broadcasts it in a specific round (say the first round of the
protocol). From the assumptions on synchrony all correct
nodes can check whether a value was indeed sent at the specified
round and whether multiple (faulty) values were sent. In the
transient fault model no such round number can be set beforehand
adjoined with the broadcast. Thus a faulty General has more power in
trying to fool the correct nodes by sending its values at completely
different times to whichever nodes it chooses.

The primitive \initiator aims at making the correct nodes associate
a local time with the invocation of the protocol (and primitive) by (the possibly
faulty) General, and to converge to a single candidate value for the
agreement to come. Since the full invocation of the protocol by a
faulty General might be questionable, there may be cases in which
some correct nodes will return a $\perp$ value and others will not
identify the invocation as valid. But, as we will prove, if any correct node happens to return a
value $\neq\perp$ within a given timeframe, all correct nodes will return the same value.

In order to initiate the process of broadcasting its value (one value at a time)  the General sends (Initiator$,G,m$) to all nodes, 
provided some validity criteria are met, as we detail below.  
As a response to that initiation message, each non-faulty node (including the General)  invokes the   primitive \initiator.
Each node dynamically executes the primitive, whenever relevant messages are being received,
to obtain an estimate to its (relative) local-time at which the primitive may
have been initiated. The primitive guarantees that all correct nodes'
estimates are within some bounded real-time of each other. \\

To ensure convergence we need to add to the two Sending Validity Criteria of \sectionref{sec:byzagree} a third one:
\begin{itemize}
\item[{[IG3]}] No invocation of \initiator$(G,*)$ failed in the last  $\dreset$ time, where an invocation is considered failed if any of the following is failed: executing
lines L4, M4 or N4 of the \initiator primitive (see \figureref{alg:init-agree}) is not  completed within $2d$, $3d$ or $4d$  of the invocation, respectively.
%, or Line R4 of the  \ByzAgreementl protocol (see \figureref{alg:Byz-alg}) is not  completed within $4d$ of the invocation.
\end{itemize}

%:initi-primitive

\begin{figure*}[t]
\center \hspace{-1.7cm} \fbox{\begin{minipage}{5.6in} \footnotesize
\begin{alltt}
\setlength{\baselineskip}{3.5mm}

\noindent Primitive \initiator$(G,m)$
\mbox{\ }\hfill\textit{/* Executed at node $q.$ $\tau_q$ is the local-time at $q.$ */}\\
\mbox{\ }\hfill\textit{/* Block~K is executed only when (and if) the primitive is explicitly invoked. */}\\
\mbox{\ }\hfill\textit{/* Lines L1 through N3 are repeatedly executed upon receiving messages. */}\\

K1.\ {\bf if} $\broadcast[G,m']=\perp$ for every $m'\neq m$  {\large \&} $last_q(G)=\bot$ {\large \&}\\
\tre  did not send any $(support,G,*)$ in 
%the interval  
$[\tau_q-d,\tau_q]$ {\large \&}\mbox{\ } \hfill\textit{/* allow for recent messages */}\\
\tre 
 $last_q(G,m)=\bot$ 
 at $\tau_q-d$
 % the following was removed since the first condition covers it
 %it: $\broadcast[G,*]=\perp$  {\large \&} 
{\bf then}\\
K2.\ \tb $\broadcast[G,m]:=\tau_q-d$;\hfill\textit{/* recording time */}\\
\due\hspace{1mm}
send $(support,G,m)$ to all; $last_q(G,m)=\tau_q$; \\

L1.\ {\bf if} received $(support,G,m)$ from $\ge n-2f$  distinct nodes\\
\tre \tb  in the interval $[\tau_q-\alpha,\tau_q]$ for  $\alpha\le 4 d$ {\bf then}\hfill\textit{/* shortest interval */}\\
L2.\ \tb $\broadcast[G,m]:=\max\{\broadcast[G,m],\;(\tau_q-\alpha-2d)\}$; $last_q(G,m)=\tau_q$;
\hfill\textit{/* recording time */}\\
L3.\ {\bf if} received $(support,G,m)$  from $\ge n-f$ distinct nodes\\
\tre \tb  in the interval $[\tau_q-2d,\tau_q]$  {\bf then}\\
L4.\ \tb  send $(approve,G,m)$ to all; $last_q(G,m)=\tau_q$;
%dd added
\hfill\textit{/* if not recently sent */}\\

M1.\ {\bf if}
received
$(approve,G,m)$ from $\ge n-2f$ distinct nodes
\\
\tre \tb   in the interval $[\tau_q-5d,\tau_q]$  {\bf then}\\
M2.\ \tb
$\readygm=$`true'; $last_q(G,m)=\tau_q$;\\
M3.\ {\bf if}
received $(approve,G,m)$ from $\ge n-f$
distinct nodes \\
\tre \tb  in the interval $[\tau_q-3d,\tau_q]$  {\bf then}\\
M4.\ \tb
send $(ready,G,m)$ to all; $last_q(G,m)=\tau_q$;
\\

N1.\ {\bf if}
$\readygm$
{\large \&}  received $(ready,G,m)$ from $\ge
n-2f$ distinct  nodes {\bf then}\\
N2.\ \tb
send $(ready,G,m)$ to all; $last_q(G,m)=\tau_q$;\\
N3.\ {\bf if}
$\readygm$
{\large \&}  received $(ready,G,m)$ from $\ge n-f$
distinct  nodes {\bf then}\\
N4.\tb\  $\TAU_q:=\broadcast[G,m]$; $\broadcast[G,*]:=\bot$;\\
\due\hspace{1mm}  remove all $(G,m)$ messages and ignore all $(G,m)$ messages for $3d$;\\
\due\hspace{1mm}  \iacpt $\la G,m,\TAU_q\ra$;  $last_q(G,m)=\tau_q$; $last_q(G):=\tau_q$.\\

{\bf cleanup:}\\
\tb Remove  any value or message that is older than $\drmv$ time units;\\
\tb If $last_q(G)> \tau_q$ or $last_q(G)<\tau_q-(\dz-6d)$  then $last_q(G):=\perp.$\\
\tb If $last_q(G,m)> \tau_q$ or $last_q(G,m)<\tau_q-(2\drmv+9d)$ then $last_q(G,m):=\perp.$
\end{alltt}
\normalsize
\end{minipage} }
\caption{The primitive \initiator} \label{alg:init-agree}
\end{figure*}

The General, before initiating the primitive, removes from its memory all previously received messages associated with any previous invocation of the primitive with him as a General.

Each correct node records the local-time at which it receives each
message associated with the invocation of the primitive, for the specific General.  Whenever a new message arrives the node records it and its time of arrival. The node goes through the primitive and considers all the various lines of the primitive, one by one, and acts accordingly.  Notice that the node processes all messages, even if it did not invoke the primitive.

We say that a node does an {\bf \iacpt } of a value sent by the
General if it accepts this value as the General's initial value, and
$\TAU_q$ is the estimated local-time at $q$ associated with the
initiation of the primitive by the General.

Each node maintains a list $\broadcast[G,*]$ for the possible concurrent
values sent by the General G, where each non-empty entry is a
local-time associated with the possible invocation of the primitive with that entry value. 
The list should contain at most a single value if the General it correct.
Each node also maintains for each non empty entry a time variable, $\lastgm,$ that indicates 
the latest time at which any stage of the primitive was executed regarding the specific value 
$m$. To ensure the compliance of the General with the rules of initiating the primitive each 
node also maintains an additional time variable, $last(G)$, measuring the minimal time 
between two consecutive invocations of the primitive by the General.

Each entry has an expiration time, and messages have a decay time, so after some time all residue of previous invocations are removed. The variables are set to $\perp$ as a result of resetting them.
The \initiator
primitive requires the knowledge of the state of the vector
$\broadcast[G,*]$ $d$ time units in the past. It is assumed that
the data structure reflects that information.

\begin{definition}
The
data structure of a node is {\em fresh} with respect to a value $m$ if  $d$ units of time ago
$\broadcast[G,*]$ did not contain any value and the time variables $\lastgm,$ and $last(G)$ both were $\perp$.
\end{definition}

Thus, as we prove later on, when the data structure is fresh and a correct node receives an initiation message form a correct $G$ it will be able to execute successfully Block~K of the  \initiator primitive.

Before stating the properties that the primitive \initiator satisfies we give some intuition regarding it.  The primitive is composed of five sections: four of them are commands to execute in response of receiving messages and the final one is a cleanup process that is carried on in the background.   

\underline{Block~K}  states the rules for the invocation of the primitive. It is executed as a result of receiving a  (Initiator$,G,m$) message from $G.$ 

Line~K1 lists the tests a node carries to ensure that $G$ respects the Sending Validity 
Criteria. 
The nodes tests whether any other broadcasts of messages were processed not too long 
ago. Since the message from $G$ may take $d$ to arrive, and responses to such a message 
from other correct nodes may have been received already. Therefore the node checks what 
was the status of its data structure $d$ time units ago.  It checks whether it recently 
responded to any initiation message or whether it processed the relevant  message from 
other nodes only in the last $d$ units of time. 

Line~K2: the node sends its support message to all nodes, and marks the time of sending. The sending event entry is marked as a time prior to the invocation of the primitive, therefore $d$ is reduced.

\underline{Block~L} intends to capture the fact that enough correct nodes have sent the support messages within a short period of one another.  If that happens an approved message is being produced.

Line~L1:  The node tests whether at least one correct node has invoked the primitive in the last 4d time units.

Line~L2: The node marks that latest such event.  The node reduces $2d$ to mark a time prior to an invocation event would $G$ was a correct node.

Line~L3: The node checks whether at least $t+1$ correct nodes have sent support within $d$ of each other.  Notice that since some messages may take $0$ time to arrive and some may take $d$ the interval is $2d$.  Notice also that if at some correct node this is true, at all correct nodes the test of Line~L1 is true.

Line~L4: Since the node knows that every correct node will end up executing Line~L2, and $d$ after that all will have Line~L3 enabled, it is safe to send an approve message.\\

The general controls the previous blocks be deciding when to send the invocation messages to which correct node.  We now moved to two stages that are controlled by the correct nodes that send the approval messages, and there is a need to prevent transient messages that may happen to be in the memory of the correct nodes from separating the agreement among correct nodes.

\underline{Block~M} intends to verify that all correct nodes have moved a stage before we move to the acceptance stage.  If enough correct have sent approve within a small time window a ready message will be produced.

Line~M1: The node tests whether at least one correct node has sent a recent approve.

Line~M2: In such a case, the correct node marks that by setting the ready variable, which will mark its potential readiness to move to the final stage and to to join others in Line~N2.

Line~M3: The node checks whether every correct node will notice the sending of an approve message. 

Line~M4: In such a case the node sends a ready message and move to the final stage.

\underline{Block~N} is the only block that is not timed by a short interval, in order to enable nodes that may be initial spread around to collect their actions.  If enough have noticed the readiness to accept the message by the general, all will.

Line~N1: The node tests whether at least one correct node has sent a ready message and whether it is ready to move to the final stage.

Line~N2: In such a case, the correct node amplifies the sending and sends its own ready message.

Line~M3: The node checks whether every correct node will notice the sending of ready messages.

Line~M4: in such a case the node set the potential time of the invocation of the protocols by $G$ and accepts the sending.  In order to prevent recurrence accepting the node clear messages and ignore messages for a short time period.

\underline{Block~Cleanup} has three parts. Any message that is too old is removed.  The other two parts rest the two variables that measure the elapse time between two consecutive invocations of the same value and of different values. The reason that the expiration of $last_q(G,m)$ is almost twice $\drmv$ is to separate consecutive sending of the same value from the possible transient messages at startup.   

Recall that a node is required to keep time stamps associated with the various entries in its data structures and the messages it has received. 
Each time-stamped
entry that is clearly wrong, with respect to the current clock reading of $\tau_q$,
is removed;  \ie future time stamps or too old time stamps.\\

%:init properties
The primitive \initiator satisfies the following properties, provided that the system is stable:
\begin{enumerate}
\item[\small\bf {[IA-1]}] \emph{(Correctness)} 
%dd changed
%If all correct nodes invoke
%\initiator$(G,m),$ with {\em fresh} data structures, within some real-time interval $[t_0,t_0+d]$,
If a correct General $G$  invokes \initiator$(G,m)$ at $t_0$
 then:
\begin{enumerate}\item[\small\bf {[1A]}]All correct nodes \iacpt~$\la G,m,\TAU\ra$ within $4d$
time units of the invocation; 
\item[\small\bf {[1B]}]
All correct nodes \iacpt$\la G,m,\TAU\ra$ within $2d$ time units
of each other;
\item[\small\bf {[1C]}] For every pair of correct nodes $q$ and $q^\prime$ that \iacpts $\la G,m,\TAU_q\ra$
and $\la G,m,\TAU_{q^\prime}\ra,$ respectively:
$$|rt(\TAU_{q^\prime})-rt(\TAU_q)|\le d\,;$$
\item[\small\bf {[1D]}] For each correct node $q$ that \iacpts $\la
G,m,\TAU_q\ra$ at $\tau_q,$ $t_0-d\le rt(\TAU_q) \le rt(\tau_q)\le
t_0+4d.$
\end{enumerate}
  \item[\small\bf {[IA-2]}] {\em (Unforgeability)} If no correct node invokes \initiator$(G,m),$ then no correct node 
\iacpts $\la G,m,\TAU\ra.$

  \item[\small\bf {[IA-3]}] {\em ($\dagr$-Relay)} If a correct node $q$ \iacpts  $\la
  G,m,\TAU_q\ra$ at real-time $t$, such that
%dd changed since it is clear
%  $0\le t-rt(\TAU_q)\le\dagr,$
  $t-rt(\TAU_q)\le\dagr,$
   then:
\begin{enumerate}\item[\small\bf {[3A]}] Every correct node $q^\prime$ \iacpts  $\la
G,m,\TAU_{q^\prime}\ra,$ at some real-time $t^\prime,$ with
  $|t-t^\prime|\le2 d$ and
  %dd changed
%   $|rt(\TAU_q)-rt(\TAU_{q^\prime})|\le 5d.$ 
   $$|rt(\TAU_q)-rt(\TAU_{q^\prime})|\le 6d\,;$$
   %dd changed since it is hard to define cleanly
  %   \item[\small\bf 3B] Moreover, $rt(\tau^G_q), rt(\tau^G_{q^\prime}) \in [t_1-2d,t_2],$ where
%$[t_1,t_2]$ is the interval within which all correct nodes that
%actually invoked 
  \item[\small\bf {[3B]}] Moreover, for every correct node $q^\prime$,  $rt(\tau^G_{q^\prime}) \le t_2,$ where some correct node invoked the primitive \initiator at $t_2;$
  \item[\small\bf {[3C]}]For every correct node $q^\prime,$$\;rt(\TAU_{q^\prime})\le rt(\tau_{q^\prime})$
  and $rt(\tau_{q^\prime})-rt(\TAU_{q^\prime})\le \BYZdur+8d.$
  \end{enumerate}
   \item[\small\bf {[IA-4]}] {\em (Uniqueness)} If a correct node $q$ \iacpts  $\la
  G,m,\TAU_q\ra,$ and a correct node \iacpts  $\la
  G,m^\prime,\TAU_p\ra$, then
 \begin{enumerate}
 \item[\small\bf {[4A]}]   
  for $m\neq m^\prime,$ 
  %for
%  $|rt(\TAU_q)-rt(\TAU_p)|\le 5d.$
% by increasing \dz-6d we can increase this difference
    $|rt(\TAU_q)-rt(\TAU_p)|> 4d;$ 
  \item[\small\bf {[4b]}]   
  for $m= m^\prime,$ either
  %for
%  $|rt(\TAU_q)-rt(\TAU_p)|\le 5d.$
% by increasing \dz-6d we can increase this difference
    $|rt(\TAU_q)-rt(\TAU_p)|\leq 6d$ or $|rt(\TAU_q)-rt(\TAU_p)|>2\drmv-3d$.
   \end{enumerate} 
\end{enumerate}

%\dd{another option is to send a value at the end of the primitive and to remove from DS xd after reception}

When the primitive is invoked the node executes Block~K.
A node may receive messages related to the primitive, even in case
that it did not invoke the primitive. In this case it
executes the rest of the blocks of the primitive, if the appropriate
preconditions hold. A correct node repeatedly executes each line until it execute Line~N4.  So we assume that a node may send the same message several times. We ignore possible optimizations that can save such repetitive sending of messages.  
Once a node executes Line~N4 it removes all associated messages and ignores related messages for some time, so Line-N4 is not executed more than once per execution of the primitive.

Notice that since Block~N is not timed, its expiration is determined by the expiration of old messages, which leads to the definition of  $\drmv.$ and $\dv.$ 
Following the completion of \ByzAgreement, the
data structures of the related \initiator instance are reset.

\looseness -1 
The proof that the \initiator\ primitive satisfies the {\small
[IA-*]} properties, under the assumption that $n>3f,$ appears in
\sectionref{ssec:inagree}. The proofs also show that from any initial state, after $\dstb$ the system becomes stable.

\section{The \broadcastpl Primitive}\label{sec:reliable-bdcst}
This section presents the \broadcastp (a message driven broadcast)
primitive, which {\bf accept}s messages being
{\bf broadcast}ed.
The primitive is invoked within
the \ByzAgreement protocol presented in \sectionref{sec:byzagree}.
The primitive follows the broadcast primitive of Toueg, Perry, and
Srikanth~\cite{FastAgree87}. In the original synchronous model,
nodes advance according to rounds that are divided into phases. This
intuitive lock-step process clarifies the presentation and
simplifies the proofs. Here the primitive \broadcastp is presented
without any explicit or implicit reference to absolute time or round number, rather an anchor
to the potential initialization point of the protocol is passed as a
parameter by the calling procedure. The properties of the \initiator
primitive guarantee a bound between the real-time of the anchors of
the correct nodes. Thus a general notion of a common round structure
can be implemented by measuring the time elapsed since the anchor.

In the broadcast primitive of~\cite{FastAgree87} messages associated
with a certain round must be sent by correct nodes at that round and
will be received, the latest, at the end of that round by all
correct nodes. In \broadcastp, on the other hand, the rounds
progress with the arrival of the anticipated messages. Thus for
example, if a node receives some required messages before the end of
the round it may send next round's messages. The length of a round
only imposes an upper bound on the acceptance criteria. Thus the
protocol can progress at the speed of message delivery, which may be
significantly faster than that of the protocol
in~\cite{FastAgree87}.

%:rel-bdcst protocol
\begin{figure*}[t]
\center \fbox{\begin{minipage}{5in} \footnotesize
\begin{alltt}
\setlength{\baselineskip}{3.5mm}


\noindent Primitive \broadcastp $(p,m,k)$\\
\mbox{\ }\hfill\textit{/* Executed per such triplet at node q. */}\\
\mbox{\ }\hfill\textit{/* Nodes send specific messages only once. */}\\
\mbox{\ }\hfill\textit{/* Nodes execute the blocks only when $\TAU$ is defined. */}\\
\mbox{\ }\hfill\textit{/* Nodes log messages until they are able to process them. */}\\
\mbox{\ }\hfill\textit{/* Multiple  messages  sent by an  individual node are ignored. */}\\

\ \tb {\bf At node } $q=p$:\hfill\textit{/* if node $q$ is node $p$ that invoked the primitive */}\\
V.\ \ \tb node $p$ sends $(init,p,m,k)$ to all nodes;\\

%\tb At any node $q$ once $\TAU_q$ is defined:\\
%dd should this be 2k-1?
W1.\ {\bf At time }$\tau_q:$  $\tau_q\le \TAU_q+2k\cdot\Phi$ \\
W2.\ \tb {\bf if} received $(init,p,m,k)$ from $p$  {\bf then}\\
W3.\ \tre send $(echo,p,m,k)$ to all;\\


%X1.\ {\bf At time }$\tau_q:$  $\tau_q\le \TAU_q+(2k-1)\cdot\Phi$\ddu \\
%dd should it be:
X1.\ {\bf At time }$\tau_q:$  $\tau_q\le \TAU_q+(2k+1)\cdot\Phi$ \\
X2.\ \tb {\bf if} received $(echo,p,m,k)$ from $\ge n-2f$
distinct nodes {\bf then}\\
X3.\ \tre send $(init^\prime,p,m,k)$ to all;\\
X4.\ \tb {\bf if} received $(echo,p,m,k)$ messages from $\ge n-f$
distinct nodes {\bf then}\\
X5.\ \tre  accept $(p,m,k)$;\\

Y1.\ {\bf At time }$\tau_q:$  $\tau_q\le \TAU_q+(2k+2)\cdot\Phi$ \\
Y2.\ \tb {\bf if} received $(init^\prime,p,m,k)$
from $\ge n-2f$ {\bf then}\\
Y3.\ \tre  $broadcasters:=broadcasters\cup\{p\}$;\\
Y4.\ \tb {\bf if} received $(init^\prime,p,m,k)$ from $\ge n-f$
distinct nodes {\bf then}\\
Y5.\ \tre  send $(echo^\prime,p,m,k)$ to all;\\

Z1.\ {\bf At any time}: \\
Z2.\ \tb {\bf if} received $(echo^\prime,p,m,k)$ from $\ge
n-2f$ distinct nodes {\bf then}\\
Z3.\ \tre  send $(echo^\prime,p,m,k)$ to all;\\
Z4.\ \tb {\bf if} received $(echo^\prime,p,m,k)$ from $\ge n-f$
distinct nodes {\bf then}\\
Z5.\ \tre  accept $(p,m,k)$;\hfill\textit{/* accept only once */}\\
\\
{\bf cleanup:}\\
\tb Remove any value or message older than $(2f+3)\cdot\Phi$
time units.
%\tb{\bf end}
\end{alltt}
\normalsize
\end{minipage} }
\caption{The \broadcastp primitive with message-driven round
structure} \label{alg:reg-bdcst}
\end{figure*}

Note that when a node invokes the primitive it evaluates all the
messages in its buffer that are relevant to the primitive. The
\broadcastp primitive is executed in the context of some initiator
$G$ that invoked \ByzAgreement, which makes use of the
\broadcastp primitive.  No correct node will execute the \broadcastp
primitive without first producing the reference (anchor),
$\TAU,$ on its local timer to the time estimate at which $G$
supposedly invoked the original agreement. By {\small IA-3A} this
happens within
%dd changed
%$2d$
$6d$
of the other correct nodes.

The synchronous Reliable Broadcast procedure of~\cite{FastAgree87}
assumes a round model in which within each phase all message
exchange among correct nodes take place. The equivalent notion of a
round in our context will be $\Phi$ defined to be: $\Phi:=
t^{\mbox{\sc\tiny G}}_{\mbox{\sc\tiny skew}}+2d.$

%:rel-bdcst properties
The \broadcastp primitive satisfies the following [TPS-*] properties
of Toueg, Perry and Srikanth~\cite{FastAgree87}, which are phrased
in our system model.
\begin{enumerate}
\item[\small\bf TPS-1] (\emph{Correctness}) If a correct node $p$
\broadcastp$(p,m,k)$ at $\tau_p$, where
$\tau_p\le\TAU_p+(2k-1)\cdot\Phi,$
%dd should it be 2k?
on its timer, then each correct
node $q$ accepts $(p,m,k)$ at some $\tau_q$,
$\tau_{q}\le\TAU_{q}+(2k+1)\cdot\Phi,$  on its timer and
$|rt(\tau_p)-rt(\tau_q)|\le 3d.$
  \item[\small\bf TPS-2] ({\em Unforgeability}) If a correct node $p$ does not
\broadcastp$(p,m,k)$, then no correct node accepts
$(p,m,k).$
  \item[\small\bf TPS-3] ({\em Relay}) If a correct node $q_1$ accepts $(p,m,k)$
  at $\tau_1,$
$\tau_1\le\TAU_1+r\cdot\Phi$ on its timer  then any other
correct node $q_2$ accepts $(p,m,k)$ at some $\tau_2,$
$\tau_2\le\TAU_2+(r+2)\cdot\Phi,$ on its timer.
  \item[\small\bf TPS-4] ({\em Detection of broadcasters}) If a correct node
accepts $(p,m,k)$  then every correct node $q$ has $p\in
broadcasters$ at some $\tau_q,$ $\tau_q\le
\TAU_q+(2k+2)\cdot\Phi,$ on its timer. Furthermore, if a correct
node $p$ does not \broadcastp any message, then a correct node
can never have $p\in broadcasters.$
\end{enumerate}

Note that the bounds in [{\small TPS-1}] are with respect to $d$,
the bound on message transmission time among correct nodes.

When the system is stable, the \broadcastp primitive satisfies the [{\small TPS-*}] properties,
under the assumption that $n>3f.$ The proofs that appear in
\sectionref{ssec:bcast} follow closely the original proofs
of~\cite{FastAgree87}, in order to make it easier for readers that
are familiar with the original proofs.

\section{Proofs}\label{sec:apdx}\ignore{-2mm}

Note that all the definitions, theorems and lemmata in this paper
hold only from the moment, and as long as, the system is stable.

\subsection{Proof of the \initiator Properties}\label{ssec:inagree}
In the proof we distinguish between the initiation of the primitive \initiator
by the General that is done by sending (Initiator$,G,m$) to all nodes, and the
invocation of the primitive \initiator by the non-faulty nodes as a result of
receiving the above message. Notice that the General himself plays a double role;
it also invokes the primitive.

Nodes continuously run the primitive, in the sense that for each
incoming message the various ``if statements'' are tested. We say
that a node {\em executes} a line in the code when the appropriate ``if condition'' holds.
In the proofs below, we omit the reference to $(G,m)$ when it is clear from the context.
Thus, when we refer to a node executing a line it is assumed that it is with $(G,m)$ 
and that the ``if''  condition holds.

\begin{claim}\label{clm:no-sending}
If a correct General $G$ doesn't initiate \initiator in an
interval $[\bar t-\dreset, \bar t)$ then, 
\begin{enumerate}
\item
at $\bar t$
when $G$ initiates the primitive \initiator with $m$, all correct nodes 
will execute successfully Line~K1 and will send $\support$ in the interval $[\bar t,\bar t+d]$;
\item
by $\bar t+4d$ all correct nodes 
will execute Line~N4, within $2d$ of each other;
\item
at any $t'\ge \bar t$, if the correct
$G$  initiates its \initiator with value $m'$ and
$G$ did not initiate any  \initiator  in the interval $[t'-\dz,t')$ and
 $G$ did not initiate any  \initiator with $m'$ in the interval $[t'-\dv,t')$
 then
all correct nodes will execute successfully Line~K1 and will send $(support, G , m')$ 
in the interval $[t',t'+d]$, and by $ t'+4d$ all correct nodes 
will execute Line~N4, within $2d$ of each other.
\end{enumerate}
\end{claim}
\begin{proof}
Notice that $\support$ messages are sent only as a result of receiving the  initiation message from the 
General.
Recall that  $\dv=15d+2\drmv$ and $\dreset=20d+4\drmv$.
Define $t=\bar t-20d-4\drmv.$
In the proof we consider only nodes that are correct from time $t$ on.
At $t+d$ some correct nodes may still end up executing (successfully\footnote{We omit
the term ``successfully'' from now on}) Block~K and may end up sending $\support$, because of some presumably previously received messages;  but past $t+d$, by the 
code of the
primitive, no correct node would execute it any more.  The last $\support$ message resulting from that activity may reach some non-faulty node the latest by $t+2d$.  For that reason,
past $t+6d$, no correct node will execute Block~L until a new initiation message will be received by some 
correct node.

The latest $\approve$ may be sent by $t+4d$ and reach others by $t+5d.$ But past $t+10d$, no correct node 
will execute Block~M.  Notice that faulty
nodes may still influence some correct nodes to execute Block~N and it might
be that some and not all correct nodes will follow them.

By $t+10d+\drmv$ the variable $\readygm$ (for all possible values of $m$) will decay at all correct nodes
and none will execute Block~N or update $\lastgm$ anymore. By $t+10d+2\drmv+d$ no correct node
will hold in its memory any message claimed to be sent by a correct node
and all variables in all data structures, including $last_q(G)$, will decay.  The
variable $last_q(G,m)$ will decay at all correct nodes by $t+10d+2\drmv+2\drmv+9d=t+19d+4\drmv=t+\dreset-d.$ 

Therefore, if at time $\bar t$  the correct $G$ will initiate \initiator with any $m$, all
correct nodes will execute successfully Line~K1 and will send $support$
within $d$ of each other, completing the proof of the first item of the claim.

To prove the second item of the claim, notice that by $\bar t+2d$ all correct nodes will execute successfully 
Line~L4,
and by   $\bar t+3d$ all will execute successfully both lines M2 and M4.
By  $\bar t+4d$ all will execute successfully Line~N4. Let $q$ be the first correct node executing Line~N4 at 
some time $ t_1$ in this interval, following its execution of lines M4 and N3.  By $ t_1 +d$ all will execute 
Line~N2 and by $ t _1+2d$ all will execute Line~N4,  and will
set the value of $\lastgm$ and $last_q(G).$    

To prove the third item of the claim we will use a mathematical induction on the initiations of \initiator past time 
$\bar t.$ Since the correct $G$ initiates \initiator sequentially, the order of initiations is well defined. Let $i,$ $i
\ge 0$, be the index describing the order of initiations past time $\bar t.$  Case $i=0$ holds by the first two items of the claim.  

Assume that the third item holds for $i-1$ and  prove it for $i.$ 
Let $t$ be the time at which the $i-1$ initiation started.  By the induction hypothesis, and by the code of the 
primitive, by $t+4d+\dz-6d<t+\dz$ all will reset $last_q(G)$.  Therefore, by $t'$ all non-faulty have reset the value of $last_q(G)$.
If $G$ did not initiate \initiator with $m'$ after time $t'-\dv$, then by the proof of the first item of the claim, all 
will execute Block~K. The flow of the proof of the second item of the claim completes the proof.

Otherwise, let $t_0,$ $t_0\ge \bar t,$ be the last time $G$ initiated \initiator with $m'$.  By the induction 
hypothesis, by $t_0+4d$ all non-faulty nodes will execute Line~N4 and by $t_0+4d+2\drmv+9d$ all would 
have reset $last(G,m')$. Since  $t_0+4d+2\drmv+9d< t_0+\dv\le t'$, again,  all will execute Block~K, and following  the 
arguments of the first two items, the claim holds.
\end{proof}

The proof can be extended to prove the following corollary for non-faulty nodes that become correct.

\begin{corollary}\label{lcor:no-sending}
\claimref{clm:no-sending} holds for any set of at least $n-f-1$ nodes and a General that are all non-faulty  from 
time $\bar t-\dreset$ on.
\end{corollary}

In the proofs below we need to refer to the coherency of the system and to the
minimal time past from the time the network becomes correct. We denote by $\iota_0$ the
time by which the network becomes correct and there are at least $n-f$
non-faulty nodes that remain non-faulty from that time on.  The system is considered {\em stable} from time $\iota_1=\iota_0+2\dreset$, and as long as the system remains coherent.

In the rest of this section, in all the claims and proofs,  whenever we refer to a non-faulty node we imply a non-
faulty node that  remains non-faulty from time $\iota_0$ on. 

%\dd{say 2drmv+9d=2drm+dz-4d}

\begin{lemma}\label{lemma-good-reset}
Once the system is stable, at any time past time $\iota_1$, if a correct General $G$ initiates the primitive \initiator 
at some time $\hat t,$ not sooner than $\dz$ of the beginning of the previous initiation,
and not sooner than $\dv$ of the last initiation with the the same value $m$,
then within $d$ of the initiation, all correct nodes will send $(support,G,m)$. Moreover, by $\hat t+4d$ all 
correct nodes 
will execute Line~N4, within $2d$ of each other.
\end{lemma}

\begin{proof}
%dd nodes are non-faulty since not enough time have passed
Recall that $\iota_0$ is the time by which the network became correct, as defined above.
Before $\iota_0$ every non-faulty node may have arbitrary values in the various
variables of \initiator and some of the messages being accumulated may be a result of the transient fault.

Past $\iota_0+d$ all received messages claimed to be sent by non-faulty nodes were
actually sent by non-faulty nodes. Observe that messages resulting form the
initial arbitrary state may be sent by non-faulty nodes as a result of their
initial state without actually receiving the required messages, since such
messages may be in their initial memory state.

Past $\iota_0+6d$, whenever a non-faulty nodes considers $support$ or $approve$
messages that were received within the appropriate time intervals in Block~L
and Block~M of the primitive it  considers only messages from
non-faulty nodes that were sent by non-faulty nodes as a result of executing the code of \initiator.

Past $\iota_0,$ if a non-faulty General $G$ doesn't initiate \initiator in an interval $[t,t+\dreset)$, where $t+
\dreset\le \hat t,$
by  \claimref{clm:no-sending} the lemma holds.

Now assume that the non-faulty node $G$ did initiate \initiator in the interval $[\iota_0,\iota_0+\dreset)$.
If during any such invocation (when executing \initiator as one of the participating
nodes) $G$ fails to successfully execute either Line~L4 within $2d$ of the invocation, or
Line~M4 within $3d$ of the invocation or Line~N4 within $4d$ of of the invocation,
then it will not initiate the primitive for another  $\dreset$, and by
\claimref{clm:no-sending} the lemma holds.

The only case that is left is when $G$ did initiate \initiator in the
interval $[\iota_0,\iota_0+\dreset)=[\iota_0,\iota_0+20d+4\drmv)$ and whenever it does so, it successfully  
executes Line~L4, Line~M4,
Line~N4  within $2d$, $3d$, and $4d$, respectively.   Recall that before
initiating the primitive a non-faulty General removes all past messages associated
with the primitive.  Let $\bar t> \iota_0+d$ be a time at which $G$ invoked the primitive.
Therefore,  past time $\bar t$, all messages from
non-faulty nodes that $G$ receives, while executing the primitive, were actually sent by
non-faulty nodes. By assumption, by $\bar t+2d$ $G$ executes Line~L4, therefore, by the code of the primitive, by $\bar t+3d$ 
all non-faulty nodes
would have $\broadcast[G,m]$ defined. Similarly, by $\bar t+3d$  $G$ execute Line~M4, therefore by $\bar t
+4d$ all will have
$\readygm=$`true'. Since all the $(ready,G,m)$ messages $G$ accumulates were
actually sent past $\bar t-d$ (it may receive these messages after invoking the primitive), 
all non-faulty nodes will receive at least $t+1$ of them by $\bar t+5d$,
and by $\bar t+6d$ all non-faulty nodes will successfully execute Line~N4.

Let $t'$  be the first time, past $\iota_0+20d+4\drmv,$ at which $G$, as a correct node invokes the primitive 
with some $m$, assuming it didn't do so with that specific $m$ for at least $\dv=15d+2\drmv$, and for any 
other $m'$ for at least $\dz=13d.$ Let $\bar t$ be the last time $G$ invoked the primitive with that specific $m
$.  By the arguments above, by $\bar t+6d$ all non-faulty nodes would have set \lastgm, and by $\bar t+6d
+2\drmv+9d\le t'$ all would have reset it.  For similar reasons, if $t$ was the last time prior to $t'$ at which $G$ 
invoked the primitive with any value, then by $t+6d$ all would have executed Line~N4, and by $t+6d+\dz-6d=t+\dz<t'$ would 
have reset the variable $last_q(G)$. Therefore, when each correct node receives the invocation it will 
send $\support$ within $d$ of each other and by $t'+4d$ all non-faulty nodes 
will execute Line~N4, within $2d$ of each other.

To complete the proof we use mathematical induction as was done in the proof of \claimref{clm:no-sending}.
\end{proof}

\lemmaref{lemma-good-reset} and the validity criteria of initiating the primitive \initiator imply the following. 

\begin{corollary}\label{cor:good-reset}
Once the system is stable, whenever a correct General $G$ initiates the \initiator with some value $m$, the data structures at all correct nodes is fresh.
\end{corollary}

We now prove some technical claims that cover the case of a faulty General.

\begin{claim}\label{clm:m2-spacing}
If a non-faulty node executes Line~M2 (or Line~M4) with some $(G,m)$ at some time  $t,$ for $t>\iota_0+10d$,
then no non-faulty node will execute Line~M2 (or Line~M4)  with  $(G,m)$ at any $t',$ $t'\in(t+10d,t+2\drmv)$
and in the interval $t'\in(t,t+2\drmv+10d)$ there is a sub-interval of length at least $2\drmv$ during which no 
non-faulty node  executes Line~M2 (or Line~M4)  with  $(G,m)$.
\end{claim}
\begin{proof}
A non-faulty node that executed Line~M2 (or Line~M4)  with  $(G,m)$  at time $t$ has considered only 
messages sent past $\iota_0+d$ and noticed at least one message from a
non-faulty node, say $q$, that has sent $(approve,G,m)$ at some time in the interval
$[t-6d,t]$. The non-faulty node $q$ sent the $(approve,G,m)$ message as a result of
executing Line~L4  at some time $t'$ in the above interval. Since $q$ have received
$n-f$ $(support,G,m)$ messages in the interval $[t'-2d,t']$, every non-faulty node
should have noticed at least $t+1$ of these in some interval $[t'-3d,t'+d]$ and
would have executed Line~L2 in that interval.  This implies that all non-faulty nodes
have set $\lastgm$ at some time in the interval $[t-9d,t+d]$. By the protocol,
no non-faulty node will send any $(support,G,m)$ later than $t+2d$ (it allows for recent messages which 
causes it to send its support a $d$ later) until it will reset $\lastgm$, which takes $2\drmv+9d$ time. The 
earliest this will happen to any
non-faulty node is $t-9d+2\drmv+9d=t+2\drmv$.

Since no non-faulty node will send $(support,G,m)$ later than $t+2d$, no non-faulty
node will execute Line~L4 later than $t+2d+2d=t+4d$, and its message may be
received by non-faulty nodes by $t+5d.$ Therefore, Line~M2 (or M4) may still
be executed as a result of such a message as late as $t+10d$.  This implies
that no non-faulty node will execute Line~M2 (or M4) in the interval $(t+10d,t+2\drmv]$. Note that by definition  $2\drmv>10d$.

Observe that the above arguments imply that if $\bar t$ is the latest time in the
interval $[t,t+10d]$ at which a non-faulty node executes Line~M2, then
no non-faulty node will execute Line~M2 or Line~M4 earlier than $\bar t+2\drmv,$ since each non-faulty node gas set $\lastgm$ at  $\bar t - 9d$ or later.
\end{proof}

\begin{corollary}\label{cor:2m2}
If two non-faulty nodes execute  Line~M2  with  $(G,m)$ at some times  $t_1, t_2,$ respectively, for $t_1,t_2 >
\iota_0+10d$,
then either $|t_1-t_2|\le 9d$ or $|t_1-t_2|> 2\drmv.$
\end{corollary}

\begin{claim}\label{clm:m4-cluster}
If a non-faulty node executes  Line~M4 with  $(G,m)$  at some time  $t,$ for $t>\iota_0+10d$,
then no non-faulty node will execute Line~M4 in the interval $[t+8d,t+2\drmv+5d]$; and in the
interval $[t,t+2\drmv+6d]$, there is sub-interval of length $2\drmv$ during which no
non-faulty node executes either Line~M2 or Line~M4 with   $(G,m)$ .
\end{claim}
\begin{proof}
A non-faulty node that executed  Line~M4  with  $(G,m)$  at time $t$ has considered only messages sent past 
$\iota_0+d$ and noticed at least $t+1$ message from 
non-faulty nodes that were sent in the interval
$[t-4d,t]$. Each such message is a result of receiving $\support$ messages that may have been sent as early 
as $t-7d.$ Thus, all these are based on actual messages being sent past $\iota_0+d.$

Let $t$ be a time
at which a non-faulty node execute Line~M4 with  $(G,m)$. By $t+d$ all  non-faulty nodes will
set $\lastgm.$ Past $t+2d$ and until its $\lastgm$ is reset no non-faulty node
will send $\support$.  Therefore, no non-faulty node will send $\approve$ past $t+2d+2d$,
and none will execute Line~M4 past $t+4d+d+3d$ and until its $\lastgm$ is reset.
Since a non-faulty node executed Line~M4 at time $t$, the set of messages causing
it to execute Line~M4 should cause all other non-faulty node to execute Line~M2 at
some time past $t-4d.$ Thus, this is the earliest time at which some non-faulty node
may have set $\lastgm$ and will not set it later. Therefore, no non-faulty node will
execute Line~M4 in the interval  $[t+8d,t+2\drmv+5d]$.

Observe that the above arguments imply that if $\bar t$ is the latest time the
interval $[t-4d,t+8d]$ at which a non-faulty node executes Line~M4, then
no non-faulty node will execute Line~M2 or Line~M4 with   $(G,m)$  earlier than $\bar t+2\drmv+5d.$
\end{proof}

\claimref{clm:m4-cluster} implies the following.

\begin{corollary}\label{cor:2m4}
If two non-faulty nodes execute  Line~M4 with  $(G,m)$ at some time  $t_1, t_2,$ respectively, for $t_1,t_2 >
\iota_0+10d$,
then either $|t_1-t_2|\le 7d$ or $|t_1-t_2|> 2\drmv.$
\end{corollary}

\begin{claim}\label{clm:no-m2}
If no non-faulty node executes Line~M2 (or Line~M4) with $(G,m)$ in an interval $(t,t+2\drmv],$
for $t>\iota_0+\drmv$, then no non-faulty node will execute Line~N2 or Line~N4 with   $(G,m)$  in the interval 
$[t+\drmv, t'']$, where $t+2\drmv<t''$ and some non-faulty node executes Line~M4 with   $(G,m)$  
at time $t''$.
\end{claim}
\begin{proof}
Because $t>\iota_0+\drmv$, all non-faulty nodes have decayed
all messages that appeared as part of the initial state that may have not been actually sent.  Since we assume
that no non-faulty node executes Line~M2 with  $(G,m)$  in the interval $(t,t+2\drmv],$
by $t+\drmv$ all will have reset $\readygm$ and will not execute Line~N2
or Line~N4 any more, so no non-faulty node will send a new $ready$ message.
By $t+2\drmv$, all will decay  all previous
$(ready,G,m)$ messages that were sent by non-faulty nodes. From
that time on, even if some non-faulty nodes will execute Line~M2,
none will be able to execute Line~N2 until a new $(ready,G,m)$
message is produced by a non-faulty node, thus until some non-faulty node execute Line~M4 with   $(G,m)$.
\end{proof}

The proof makes use of the following simple observation.

\begin{claim}\label{clm:before-support}
At any time $t$, $t>\iota+\drmv,$ if a non-faulty node sets $\broadcast[G,m]$, then some non-faulty node has sent $(support,G,m)$ later than $rt(\broadcast[G,m]).$
\end{claim}
\begin{proof}
If the node didi it in Line~K2, then it trivially holds. Otherwise, the time window considered in Line~L2 includes a sending event of a correct node, and that happened at the earliest $d$ time units before the time window span.
\end{proof}

Using the above claims we can now prove the following.

\begin{lemma}\label{lemma-faulty-reset}
Once the system is stable, if any correct node, say $q$, executes Line~N4 with  $(G,m)$,
at some time $\bar t$, where $\bar t-rt(\TAU_q)\le \drmv-9d$, then 
\begin{enumerate}
\item
all correct nodes will execute Line~N4 with  
$(G,m)$ 
within $2d$ of each other in the interval $[\bar t-2d,\bar t+2d]$;
\item for any correct node $p$, $|rt(\TAU_q)-rt(\TAU_p)|\le 6d;$
\item some correct node executed Line~M4 later than  $\bar t- \drmv+7d$
\end{enumerate}
\end{lemma}

\begin{proof}
Let $q$ be such a correct node.  By the condition in Line~N3, $\readygm$ was last set by $q$ while 
executing Line~M2 at some time $t'$, later than $\bar t-\drmv.$ Consider the interval $(\iota_0+\drmv,\bar t-
\drmv-9d).$ By the definition of stability  it is longer than $4\drmv.$ 
If no correct node executed Line~M4 (with $G,m$) in this interval, since the system is stable, then the 
preconditions of \claimref{clm:no-m2} hold.  

Otherwise, let $t_1$ be the latest time in the above interval at which a correct node executed Line~M4.  By definition 
$|t'-t_1|>9d.$ Therefore, by \corollaryref{cor:2m4} and  \corollaryref{cor:2m2}, $|t'-t_1|>2\drmv$ and this holds 
for any other correct node that executed Line~M4 or Line~M2 within $9d$ of $q$, \ie within $9d$ of $t'.$
Therefore, again, the preconditions of \claimref{clm:no-m2} hold.

By \claimref{clm:before-support},
% The assumption that $\bar t-rt(\TAU_q)\le \drmv-9d$ implies that 
some correct node have sent $\support$ 
%in the interval $[\bar t-rt(\TAU_q)-d,\bar t-rt(\TAU_q)]$. 
in the interval $[rt(\TAU_q),\bar t].$  
By the code of the primitive, it would have not done so if 
any correct node would have executed Line~M2 or Line~M4 in the interval 
$[\bar t-\drmv,rt(\TAU_q)-2d],$ since it would have set its $\lastgm$ at least $d$ prior to that sending.

This implies that $t'\ge rt(\TAU_q)-2d,$ and that any correct node executing Line~M2 or Line~M4 within 
$9d$ of $t'$ should do so later than  $t_2,$ where  $t_2=rt(\TAU_q)-2d\ge\bar t-\drmv+7d.$

By \claimref{clm:no-m2}, some correct node executed Line~M4, in the interval 
$[\bar t-\drmv-9d,\bar t].$ Since it should be within $9d$ of $t'$, by the above argument, that should happen at some time $t_3$ in the 
interval $[t_2,\bar t].$ 
Proving the third item of the claim.
By the code of the primitive, every correct node should execute Line~M2 in the interval 
$[t_3-5d,t_3+d]$. This implies that they should do so in the interval $[t_2,\bar t+d]$, which implies within the 
interval $I=[\bar t-\drmv+7d, \bar t+d]$. 

The correct node $q$ executed Line~N4 at time $\bar t.$ It has received at least $t+1$ $\ready$ messages 
from correct nodes.  Any correct node sending such a message should have executed Line~M2 prior to 
sending the message; and such a message is a result of executing either Line~M4 or Line~N2. By 
\claimref{clm:no-m2} that can happen either before time $t_1+\drmv$ or later than time $t_2$. If it would be 
earlier than $t_1+\drmv$, node $q$ would have decayed that message from its memory since we already 
argued that  $|t'-t_1|>2\drmv$.  

We conclude that all such messages from correct nodes were sent past time $t_2$. Therefore, by $\bar t+d$ 
each correct node would execute Line~N2, since its pre-conditions holds, and by $\bar t+2d$ all will execute 
Line~N4. Let $q'$ be the first correct node to execute Line~N4 past time $t_2.$ The above arguments imply 
that it has done so in the interval  $[\bar t-2d,\bar t+2d]$, and that all correct nodes would have executed 
Line~N4 within $2d$ of $q'.$ Proving the first item of the claim.

From the above discussion, some correct  node $q^{\prime\prime}$ 
executed Line~M4 in the interval $[\bar t-\drmv+6d,\bar t].$ 
Denote that time by $t^{\prime\prime}$. 
Node $q^{\prime\prime}$ collected $n-f$ $approve$ messages in the interval  
$[t^{\prime\prime}-3d,t^{\prime\prime}]$. At least one of which is from a correct node. 
Let $q'$ be that node and let $t'$ be the time it sent its $approve$ message. 
From the above discussion, $t'\in[\bar t-\drmv+6d-3d-d,\bar t]=[\bar t-\drmv+2d,\bar t].$ Node $q'$ collected 
$n-f$ $support$
messages, with at least $n-2f$ from correct nodes. Let $t_1$ be
the time at which the $(n-2f)^{\mbox th}$ $support$ message sent by
a correct node was received by $q'$. Since $q'$ executed Line~L4, all
these messages should have been received in the interval
$[t_1-2d,t_1]$. Node $q'$ should have set a recording time
$\tau,$ $rt(\tau)\ge t_1-4d,$
as a result of (maybe repeating) the execution of
Line~L2.

Every other correct node should have received the $(n-2f)^{\mbox th}$ $support$ message sent by
a correct node at some time in the interval $[t_1-d,t_1+d]$ with the set of $(n-2f)$
$support$ messages sent by correct nodes being received in the interval
$[t_1-3d,t_1+d].$ Each such correct node  should have set the recording time after (maybe
repeatedly) executing Line~L2, since this window satisfies the
precondition of Line~L1. Thus, eventually all recording times are
$\ge t_1-5d.$ Observe that since this interval is short, none of these messages 
would have been decayed by the time they are processed by the correct nodes.

Some correct node may send a $support$ message, by executing
Line~K2, at most $d$ time units after receiving these
$n-2f$ messages. This can not take place later than $t_1+2d,$
resulting in a recording time of $t_1+d,$ though earlier than its time
of sending the $support$ message. This $support$ message (with the
possible help of faulty nodes) can cause some correct node to
execute Line~L2 at some later time. The window within which the
$support$ messages at that node are collected should include the
real-time $t_1+3d,$ the latest time any $support$ from any correct
node could have been received. Any such execution will result in a
recording time that is $\le t_1+3d-2d=t_1+d.$ Thus the range of
recording times for all correct nodes (including $q$) are
$[t_1-5d,t_1+d].$  

To complete the proof of the second item we need to show that each correct node, $p$, actually sets its $\TAU_p.$
By assumption, $\bar t-rt(\TAU_q)\le \drmv-9d$, therefore $rt(\TAU_q)\ge  \bar t- \drmv+9d,$ 
This implies that $t_1+d\ge  \bar t- \drmv+9d.$ Implying that $t_1-5d\ge \bar t- \drmv+3d.$ 
Therefore, when each correct node executes Line~N4, its $\TAU$ is well defined, since the $\broadcast[G,m]$ entry wasn't decayed yet.
Thus, completing the proof.
\end{proof}

We are now ready to prove the properties of the primitive \initiator.
%:init-theorem
\begin{theorem}\label{lem:initiator}
Once the system is stable, the primitive \initiator presented in \figureref{alg:init-agree}
satisfies properties {\small [IA-1]} through {\small [IA-4]}.
\end{theorem}

\begin{proof}\mbox{\ }\\
\noindent{\bf Correctness:} \corollaryref{cor:good-reset} proves that when a correct General initiates the primitive, the data-structures at correct nodes are fresh.
Assume that within $d$ of each other all
correct nodes invoke \initiator($G,m$). Let $t_1$ be the real-time
at which the General invokes its copy of the \initiator then by $t_2,$ $t_2\le t_1+d,$ 
the last correct node did so. Since all data structures are {\em
fresh}, then no value $\{G,m^\prime\}$ appeared in  $\broadcaster$
$d$ time units before that, thus Line~K1 will hold for all correct nodes.
Therefore, every correct node sends $(support,G,m).$ Each such
message reaches all other correct nodes within $d$. Thus, between
$t_1$ and $t_2+d$ every correct node receives $(support,G,m)$ from
$n-f$ distinct nodes and sends $(approve,G,m).$ By $t_2+2d$ 
every correct node sends $(ready,G,m),$ and by $t_2+3d$
\iacpts
$\la G,m,\tau^\prime\ra,$ for some $\tau^\prime,$ thus, proving
{\small [IA-1A]}.

To prove {\small [IA-1B]}, let $q$ be the first to \iacpt after
executing Line~M4. Within $d$ all correct nodes will execute
Line~M2, and within $2d$ all will \iacpt.

Note that for every pair of correct nodes $q$ and $q^\prime$, the
associated initial recording times $\tau$ and $\tau^\prime$ satisfy
$|\tau-\tau^\prime|\le d.$ Line~K2 implies that the recording times
of correct nodes can not be earlier than $t_1-d.$ Some correct node
may see $n-2f,$ with the help of faulty nodes as late as $t_2+2d.$
All such windows should contain a $support$ from a correct node, so
should include real-time $t_2+d,$ resulting in a recording time of
$t_2-d.$ Recall that $t_2\le t_1+d,$ proving {\small [IA-1C]}.

To prove {\small [IA-1D]} notice that the fastest node may set
$\tau^\prime$
to be $t_1-d,$ but may \iacpt only by $t_2+3d\le t_1+4d.$\\

\noindent{\bf Unforgeability:} \\
If no correct node invokes \initiator and will not send
$(support,G,m),$ then no correct node will ever execute L4 and will
not send $(ready,G,m)$. Thus, no correct node can accumulate $n-f$ distinct
$(ready,G,m)$ messages and
therefore will not \iacpt $\la G,m\ra$.
%dd added
Moreover, no correct will execute lines K2 or L2, and therefore
if $G$ is correct, no correct node will invoke \initiator, and no correct  will have any entry in the Initiator's data structure.\\

\noindent{\bf $\dagr$-Relay:} \\
Let $q$ be a correct node that \iacpts 
$\la  G,m,\TAU_q\ra$ at real-time $t$, such that
$0\le t-rt(\TAU_q)\le\dagr.$ It did so as a result of executing
Line~N4.  By assumption the preconditions of  \lemmaref{lemma-faulty-reset} hold, and therefore 
all correct nodes will \iacpt $\la G,m,\tau_{\bar
q}\ra$ within $2d$ of each other, in the interval $[t-2d,t+2d],$ with $\TAU$ values that are $6d$ apart.
Thus, proving  {\small [IA-3A]}.

To prove {\small [IA-3B]} notice that any range of messages considered in Line~L2
includes a $support$  of a correct node. The resulting recording
time will never be later than the sending time of the $support$
message by that correct node, and thus by some correct node.

%dd hidden since it is complicated to define the left hand part
\hide{
To
prove the second part of {\small [IA-3B]} consider again node $\bar
q$ from the proof of the $\dagr-relay$ property.  It collected
$n-2f$ $support$ messages from correct nodes in some interval $[\bar
t_1,t_1],$ where $\bar t_1\ge t_1-2d.$ These messages, when received
by any correct node will be within an interval of $4d,$ with the
first message in it from a correct node. These messages will trigger
a possible update of the recording time in Line~L2. Thus, the
resulting recording time of any correct node cannot be earlier than
some $2d$ of receiving a $support$ message from a correct node, thus
not earlier than $2d$ of sending such a message.
} %end hide

The first part of {\small [IA-3C]} is immediate from Line~L2 and
Line~K2. For the second part observe that for every other correct
node $q^\prime,$ $rt(\tau_{q^\prime})\le rt(\tau_q) +2d$ and
$rt(\TAU_{q^\prime})\ge rt(\TAU_q)-6d$. Thus,
$rt(\tau_{q^\prime})-rt(\TAU_{q^\prime})\le
rt(\tau_q)-rt(\TAU_q)+8d\le\dagr+8d.$\\

\noindent{\bf Uniqueness:} \\
To prove {\small [IA-4]} observe that the conditions in Line~K1 implies that each non-faulty node sends a $support$ for
a single $m$ at a time. In order to \iacpt, a correct node needs to send
$approve$ after receiving $n-f$ $support$ messages.  That can happen
for at most a single value of $m,$ because $n>3f.$ 

By \lemmaref{lemma-faulty-reset}, once a correct node execute Line~N4, all do it within $2d$.  By the protocol, once a node decides it removes accepted messages and ignores new message associated with $(G,m)$ for $3d$. Therefore, all correct nodes issue \iacpt, and stop sending messages associated with $(G,m)$ before a correct one agrees to consider such messages. So past messages cannot be used again to reproduce another wave of decisions, unless a new correct node sends a new support for $(G,m).$

Previously sent messages for another value of $m$ will not produce a wave of decisions unless a new correct node will send $\support$ for such a value.  None will send support for a new value for $\dz-6d>6d,$ so by the time such a message will be sent, old values will be out of any window of consideration for executing any L or M lines of the code by a any correct node. Line N cannot be executed unless some correct node excuses Line~M4.

What is left to prove, is that
future invocations of the primitive will not violate {\small
[IA-4]}.  

Again, by \lemmaref{lemma-faulty-reset}, once a correct node execute Line~N4, all do it within $2d$. Let $q$ be the first to excute Line~N4 in the current execution, and let it be at time $rt(\tau_p)=t'.$ By $t'+d$ all non-faulty would execute Line~L2, and the latest any correct will execute Line~K2 is $t'+d$.  By inspecting the possible scenarios one can see that no non-faulty will execute Line~L2 later than $t'+5d$, and the latest value set by any correct node in that interval will never be later than $t'+d$. Thus, for every correct node $q$,  $rt(\TAU_q)\le rt(\tau_p)+d.$

The earliest time at which any correct node will send $\support$ later than that time will be at $rt(\tau_p)+\dz-6d.$ By inspecting the protocol, the earliest possible setting of value in Line~K2 will be to $rt(\tau_p)+\dz-6d-2d.$  Therefore, if we denote by $\tau$ timings in the former
invocation and by $\bar\tau$ timings in the later one, we conclude that for any two correct nodes $p$ and $q$, $rt(\bar\TAU_q)-rt(\TAU_p)\ge\dz-9d=4d.$ 
\end{proof}

We can now state the concluding corollary.

\begin{corollary}
The system converge from any initial state within $2\times\dreset=d,$ provided that there are $n-t$ non-faulty nodes that are continuously non-faulty during that period.
\end{corollary}

\begin{proof}
Since all properties hold once the system is stable, and stability is defined as 
 $2\times\dreset$ form the time the network is correct, we conclude the proof.
\end{proof}

One can reduce the requirement of having the same non-faulty nodes stay continuously so, but we do not see this optimization as an important issue. Moreover,  the proofs above shows that once a non-faulty node discards old values it can be considered correct.  Therefore we can state the following corollary.

\begin{corollary}
Once the system is stable, a non-faulty node that is non-faulty for $\dnode$ time, can be considered correct.
\end{corollary}

\subsection{Proof of the \broadcastpl Properties}\label{ssec:bcast}

%:tmp

The proofs
essentially follow the arguments in the original
paper~\cite{FastAgree87}.

%\dd{should we state the properties of RB with respect to assumed bound
%on the latest time such msgs are being sent out?}

\begin{lemma}\label{lemma-1phase}
If a correct node $p_i$ sends a message at local-time $\tau_i,$
$\tau_i\le \TAU_i+r\cdot\Phi$ on $p_i$'s timer it will be
received and processed by each correct node $p_j$ at some local-time
$\tau_j,$ $\tau_j\le \TAU_j+(r+1)\cdot\Phi,$ on $p_j$'s timer.
\end{lemma}

\begin{proof}
Assume that node $p_i$ sends a message at real-time $t$ with
local-time $\tau_i\le \TAU_i+r\cdot\Phi.$ Thus,
$\tau_i\le\TAU_i+r( \tskew +2d)\addrho.$ It should arrive at any
correct node $p_j$ within $d\addrho$. By
{\small IA-3A}, $\TAU_j$ will be defined and the message will be
processed no later than by another $d\addrho$. By {\small IA-3A},
$|rt(\TAU_i)-rt(\TAU_j)|< \tskew \addrho.$ Thus, $rt(\TAU_i)\le
rt(\TAU_j) + \tskew \addrho$, and at time $rt(\tau_j),$ by which the
message arrived and processed at $p_j,$ we get
$$rt(\tau_j)\le rt(\tau_i)+2d\addrho\le
rt(\TAU_i)+r( \tskew +2d)\addrho+2d\addrho\,,$$ and therefore
$$rt(\tau_j)\le
rt(\TAU_j)+ \tskew \addrho+r( \tskew +2d)\addrho+2d\addrho\le
rt(\TAU_j)+ (r+1)\cdot\Phi\,.$$
\end{proof}

\begin{lemma}\label{lemma4.1}
If a correct node ever sends $(echo^\prime,p,m,k)$ then at least one
correct node, say $q^\prime,$ must have sent
$(echo^\prime,p,m,k)$ at some local-time $\tau_{q\prime},$
$\tau_{q\prime}\le\TAU_{q\prime}+(2k+2)\cdot\Phi.$
\end{lemma}

\nhide{
\begin{proof}
Let $t$ be the earliest real-time by which any correct node $q$
sends the message $(echo^\prime,p,m,k).$ If $t>
rt(\TAU_q)+(2k+2)\cdot\Phi,$ node $q$ should have received
$(echo^\prime,p,m,k)$ from $n-2f$ distinct nodes, at least one of
which from a correct node, say $q^\prime,$ that was sent prior to
local local-time $\TAU_{q\prime}+(2k+2)\cdot\Phi.$ \mbox{\ }
\end{proof}
}

\begin{lemma}\label{lemma4.2}
If a correct node ever sends $(echo^\prime,p,m,k)$ then $p$'s
message $(init,p,m,k)$ must have been received by at least one
correct node, say $q^\prime,$ at some time $\tau_{q\prime},$
$\tau_{q\prime}\le\TAU_{q\prime}+2k\cdot\Phi.$
\end{lemma}

\nhide{
\begin{proof}
By \lemmaref{lemma4.1}, if a correct node ever sends
$(echo^\prime,p,m,k),$ then some correct node $q$ should send it
at  local-time $\tau_q,$ $\tau_q\le\TAU_q+(2k+2)\cdot\Phi.$ By
the primitive \broadcastp, $q$ have received
$(init^\prime,p,m,k)$ from at least $n-f$ nodes by some local-time
$\tau_q,$ $\tau_q\le\TAU_q+(2k+2)\cdot\Phi.$ At least one of
them is a correct node $q^{\prime\prime}$  who have received $n-2f$
$(echo,p,m,k)$ at some local-time $\tau_{q{\prime\prime}},$
$\tau_{q{\prime\prime}}\le\TAU_{q{\prime\prime}}+(2k+1)\cdot\Phi.$
One of which was sent by a correct node $\bar q$ that should have
received $(init,p,m,k)$ before sending $(echo,p,m,k)$ at some
local-time $\tau_{\bar q},$ $\tau_{\bar q}\le\TAU_{\bar
q}+2k\cdot\Phi.$
\end{proof}
}

\begin{lemma}\label{lemma-3d}
If a correct node $p$ invokes the primitive \broadcastp $(p,m,k)$
at real-time $t_p,$ then each correct node $q$ accepts $(p,m,k)$ at
some real-time $t_q,$ such that  $|t_p-t_q)|\le 3d.$
\end{lemma}

\nhide{
\begin{proof}
The $init$ message of $p$ sent in Line~V will arrive to every node
by $t_p+d.$ By {\small IA-3A}, by $t_p+2d$ all will have their
$\TAU$ defined and will process the $init$ message. By
\lemmaref{lemma-1phase}, all will execute Line~W3 by that time. By
$t_p+3d$ all will execute Line~X5 and will accept.
\end{proof}
}

\begin{theorem}\label{thm4.1}
The \broadcastp primitive presented in \figureref{alg:reg-bdcst}
satisfies properties [TSP-1] through [TSP-4].
\end{theorem}

\begin{proof}

{\bf Correctness:} Assume that a correct node $p$
\broadcastp\hspace{-0.35em}s $(p,m,k)$ at $\tau_p,$
$\tau_p\le\TAU_p+(2k-1)\cdot\Phi,$ on its timer. Any correct
node, say $q,$ receives $(init,p,m,k)$ and sends $(echo,p,m,k)$ at
some $\tau_q,$ $\tau_q\le\TAU_q+2k\cdot\Phi$ on its timer. Thus,
any correct node, say $\bar q$ receives $n-f$ $(echo,p,m,k)$ from
distinct nodes at some $\tau_{\bar q},$ $\tau_{\bar q}\le\TAU_{\bar
q}+(2k+1)\cdot\Phi,$ on its timer and accepts $(p,m,k).$ The
second part of the correctness is a result of \lemmaref{lemma-3d}.

{\bf Unforgeability:} If a correct node $p$ does not broadcast
$(p,m,k),$ it does not send $(init,p,m,k),$ and no correct node will
send $(echo,p,m,k)$ at some $\tau,$ $\tau\le\TAU+2k\cdot\Phi,$
on its timer. Thus, no correct node accepts $(p,m,k)$ by
$\TAU+(2k+1)\cdot\Phi$ on its timer. If a correct node would
have accepted $(p,m,k)$ at a later time it can be only as a result
of receiving $n-f$ $(echo^\prime,p,m,k)$ distinct messages, some of
which must be from correct nodes. By \lemmaref{lemma4.2}, $p$
should have sent $(init,p,m,k),$ a contradiction.

{\bf Relay:} The delicate point is when a correct node issues an
accept as a result of getting echo messages. So assume that $q_1$
accepts $(p,m,k)$ at $t_1=rt(\tau_1)$ as a result of executing
Line~X5. By that time it must have received $(echo,p,m,k)$ from
$n-f$ nodes, at least $n-2f$ of them sent by correct nodes. Since
every correct node among these has sent its message by
$\TAU+2k\cdot\Phi$ on its timer, by \lemmaref{lemma-1phase},
all those messages should have arrived to every correct node $q_i$
by $\tau_i\le\TAU_i+(2k+1)\cdot\Phi$ on its timer. Thus, every
correct node $q_i$ should have sent $(init^\prime,p,m,k)$ at some
$\tau_i,$ $\tau_i\le\TAU_i+(2k+1)\cdot\Phi,$ on its timer. As a
result, every correct node will receive $n-f$ such messages by some
$\bar\tau,$ $\bar\tau\le\TAU+(2k+2)\cdot\Phi$ on its timer and
will send $(echo^\prime,p,m,k)$ at that time, which will lead each
correct node to accept $(p,m,k)$ at a local-time $\tau_i.$

Now observe that all $n-2f$ $(echo,p,m,k)$ were sent before time
$t_1.$ By $t_1+d$ they arrive to all correct nodes. By $t_1+2d$ all
will have their $\TAU$ defined and will process them. By $t_1+3d$
their $(init^\prime,p,m,k)$ will arrive to all correct nodes, which
will lead all correct nodes to send $(echo^\prime,p,m,k).$ Thus, all
correct nodes will accept $(p,m,k)$ at time $\tau_i\le t_1+4d.$

By assumption, $t_1=rt(\tau_1)\le rt(\TAU_1)+r\cdot\Phi.$ By
{\small IA-3A}, $rt(\TAU_1)\le rt(\TAU_i)+\tskew$.  Therefore we
conclude: $rt(\tau_i)\le rt(\tau_1)+4d\le
rt(\TAU_1)+r\cdot\Phi+4d \le rt(\TAU_i)+\tskew +
r\cdot\Phi+4d\le rt(\TAU_i)+ (r+2)\cdot\Phi.$

The case that the accept is a result of executing Line~Z5 is a
special case of the above arguments.

{\bf Detection of broadcasters:} As in the original proof (\cite{FastAgree87}), we first
argue the second part.  Assume that a correct node $q$ adds node $p$
to $broadcasters.$ It should have received $n-2f$
$(init^\prime,p,m,k)$ messages.  Thus, at least one correct node has
sent $(init^\prime,p,m,k)$ as a result of receiving $n-2f$
$(echo,p,m,k)$ messages. One of these should be from a correct node
that has received the original broadcast message of $p.$

To prove the first part, we consider two similar cases to support
the Relay property. If $r=k$ and the correct node, say $q,$ accepts
$(p,m,k)$ as a result of receiving $(echo,p,m,k)$ from $n-f$ nodes
by some $\tau_q,$ $\tau_q\le\TAU_q+(2k+1)\cdot\Phi,$ on its
timer. At least $n-2f$ of them were sent by correct nodes. Since
each correct node among these has sent its message at some $\tau,$
$\tau\le\TAU+2k\cdot\Phi,$ by \lemmaref{lemma-1phase}, all
those messages should have arrived to any correct node, say $q_i,$
by some $\tau_i,$ $\tau_i\le\TAU_i+(2k+1)\cdot\Phi$ on its
timer. Thus, each correct node, say $q_j$ should have sent
$(init^\prime,p,m,k)$ at some $\tau_j,$
$\tau_j\le\TAU_j+(2k+1)\cdot\Phi,$ on its timer. As a result, by
\lemmaref{lemma-1phase}, each correct node, say $q^\prime,$ will
receive $n-f$ such messages by some $\tau_{q\prime},$
$\tau_{q\prime}\le\TAU_{q\prime}+(2k+2)\cdot\Phi$ on its timer
and will add $p$ to $broadcasters.$

Otherwise, $q$ accepts $(p,m,k)$ as a result of receiving from $n-f$
nodes $(echo^\prime,p,m,k)$ by some $\tau_q$ on its timer. By
\lemmaref{lemma4.1} a correct node, say $q_i,$ sent
$(echo^\prime,p,m,k)$ at some $\tau_i,$
$\tau_i\le\TAU_i+(2k+2)\cdot\Phi.$ It should have received $n-f$
$(init^\prime,p,m,k)$ messages by that time. All such messages that
were sent by correct nodes were sent at some $\tau,$
$\tau\le\TAU+(2k+1)\cdot\Phi,$ on their timers and should arrive
at each node $q_j,$ at some $\tau_j,$
$\tau_j\le\TAU_j+(2k+2)\cdot\Phi,$ on its timer. Since there are
at least $n-2f$ such messages, all will add $p$ to $broadcasters$ at
some $\tau,$ $\tau\le\TAU+(2k+2)\cdot\Phi,$ on
their timers. 

%\vspace{-0.4cm}
\end{proof}

%\newpage
%\input{BYZ-proof.tex}
\subsection{Proof of the \ByzAgreementl Properties}\label{ssec:jbyzagree}

\begin{theorem}\label{thm:byz} (Convergence)
Once the system is stable, any invocation of \ByzAgreement
presented in \figureref{alg:Byz-alg} satisfies the Termination
property. When $n>3f$, it also satisfies the Agreement and Validity
properties.
\end{theorem}

\begin{proof}

Notice that the General $G$ itself is one of the nodes, so if it is
faulty then there are only $f-1$ potentially faulty nodes. We do not
use that fact in the proof since the version of \ByzAgreement
presented does not refer explicitly to the General. One can adapt
the proof and reduce $\dagr$ by $2\cdot\Phi$ when specifically
handling that case.

By \corollaryref{cor:good-reset}, by the time the system becomes stable, all data structures are fresh.

We begin by proving Validity.

\noindent{{\bf Validity}:} Since all the correct nodes invoke the
primitive \ByzAgreement as a result of a value sent by a correct
$G$, they will all invoke \initiator within $d$ of each other with
fresh data structure, hence {\small [IA-1]} implies that they all will execute Block~R within $2d$ of each other, and Validity holds.\\

The rest of the proof makes use of the following two lemmata.

\begin{lemma}\label{lemma2} If a correct node $p$ aborts at
local-time $\tau_p,$ $\tau_p>\TAU_p+(2r+1)\cdot\Phi,$ on its
timer, then no correct node $q$ decides at a time $\tau_q,$
$\tau_q\ge\TAU_q+(2r+1)\cdot\Phi,$ on its timer.
\end{lemma}

\begin{proof}
Let $p$ be a correct node that aborts at time $\tau_p,$
$\tau_p>\TAU_p+(2r+1)\cdot\Phi.$ In this case it should have
identified at most $r-2$ broadcasters by that time.  By the
detection of the broadcasters property {\small [TPS-4]}, no correct
node will ever accept $\la G,m'\ra$ and $r-1$ distinct messages
$(q_i,m',i)$ for $1\le i\le r-1,$ since that would have caused each
correct node, including $p,$ to hold $r-1$ broadcasters by some time
$\tau,$ $\tau\le\TAU+(2(r-1)+2)\cdot\Phi$ on its timer. Thus, no
correct node, say $q$, can decide at a time
$\tau_q\ge\TAU_q+(2r+1)\cdot\Phi$ on its timer.
\end{proof}

\begin{lemma}\label{lemma1} If a correct node $p$ decides at time
$\tau_p,$ $\tau_p\le\TAU_p+(2r+1)\cdot\Phi,$ on its timer, then
each correct node, say $q$, decides by some time $\tau_q,$
$\tau_q\le\TAU_q+(2r+3)\cdot\Phi$ on its timer.
\end{lemma}

\begin{proof}

Let $p$ be a correct node that decides at local-time $\tau_p,$
$\tau_p\le\TAU_p+(2r+1)\cdot\Phi.$ We consider the following
cases:

\begin{enumerate}

\item $ r=0:$ No correct node can abort by a time $\tau,$ $\tau\le\TAU+(2r+1)\cdot\Phi$, since the
inequality will not hold. Assume that node $p$ have accepted $\la
G,m'\ra$ by $\tau_p\le\TAU_p+4d\le\TAU_p+\Phi.$  By the relay
property {\small [TPS-3]} each correct node will accept $\la
G,m'\ra$ by some time $\tau,$ $\tau\le\TAU+3\cdot\Phi$ on its
timer. Moreover, $p$ invokes {\bf \broadcastp}$(p,m',1)$, by the
Correctness property {\small [TPS-1]} it will be accepted by each
correct node by time $\tau,$ $\tau\le\TAU+3\cdot\Phi,$ on its
timer. Thus, all correct nodes will have $value\ne \perp$ and will
broadcast and stop by time $\TAU+3\cdot\Phi$ on their timers, when executing Block~S.

\item $1 \le r \le f-1:$ Node $p$ must have accepted $\la G,m'\ra$
and also accepted $r$ distinct $(q_i,m',i)$ messages for all $i, 2
\le i \le r$, by time $\tau,$ $\tau\le\TAU+(2r+1)\cdot\Phi,$ on
its timer. By \lemmaref{lemma2}, no correct node aborts by that
time. By Relay property {\small [TPS-3]} each $(q_i,m',i)$ message
will be accepted by each correct node by some time $\tau,$
$\tau\le\TAU+(2r+3)\cdot\Phi,$ on its timer. Node $p$ broadcasts
$(p,m',r+1)$ before stopping. By the Correctness property, {\small
[TPS-1]}, this message will be accepted by every correct node at
some time $\tau,$ $\tau\le\TAU+(2r+3)\cdot\Phi,$ on its timer.
Thus, no correct node will abort by time $\tau,$
$\tau\le\TAU+(2r+3)\cdot\Phi,$ and all correct nodes will have
$value\ne \perp$ and will thus decide by that time.

\item $r=f:$ Node $p$ must have accepted a $(q_i,m',i)$ message
for all $i, 1 \le i \le f-1$, by $\tau_p,$
$\tau_p\le\TAU_p+(2f+1)\cdot\Phi,$ on its timer, where the $f$
$q_i$'s are distinct. If the General $G$ is correct, then by Validity the claim holds. Otherwise, at least one of these $f$ nodes (which all differ from $G$), say $q_j,$
must be correct. By the Unforgeability property {\small [TPS-2]},
node $q_j$ invoked \broadcastp $(q_j,m',j)$ by some local-time
$\tau,$ $\tau\le\TAU+(2j+1)\cdot\Phi$ and decided. Since $j \le
f$ the above arguments imply that by some local-time $\tau,$
$\tau\le\TAU+(2f+1)\cdot\Phi,$ each correct node will decide.
\end{enumerate}
\vspace{-7mm}\mbox{\ }\hfill\mbox{\ }\end{proof}

\noindent \lemmaref{lemma1} implies that if a correct node
decides at time $\tau,$ $\tau\le\TAU+(2r+1)\cdot\Phi,$ on its
timer, then no correct node $p$ aborts at time $\tau_p,$
$\tau_p>\TAU_p+(2r+1)\cdot\Phi.$ \lemmaref{lemma2} implies the
other direction.\\

\noindent{{\bf Termination}:} Each correct node either terminates
the protocol by returning a value, or by time
$(2f+1)\cdot\Phi+3d$ on its clock all entries will be reset, which is a
termination of the protocol.
\\

\noindent{{\bf Agreement}:} If no correct node decides, then all
correct nodes that execute the protocol abort, and return a $\perp$
value. Otherwise, let $q$ be the first correct node to decide.
Therefore, no correct node aborts. The value returned by $q$ is the
value $m'$ of the accepted $(p,m',1)$ message. By {\small [IA-4]} if
any correct node \iacpt s, all correct nodes \iacpt with a single
value.  Thus all correct nodes return the same value.
\\

\noindent{{\bf Timeliness}:}
\begin{enumerate}
\item (agreement)
For every two correct nodes $q$ and $q^\prime$ that decide on
$(G,m)$ at $\tau_q$ and $\tau_{q^\prime},$ respectively:
\begin{enumerate}
\item
If validity hold, then $|rt(\tau_q)-rt(\tau_{q^\prime})|\le 2d,$  by
{\small [IA-3A]}; Otherwise, $|rt(\tau_q)-rt(\tau_{q^\prime})|\le
3d,$ by {\small [TPS-1]}.

\item
$|rt(\tau^G_q)-rt(\tau^G_{q^\prime})|\le 6d$ by {\small [IA-3A]}.

\item
$rt(\tau^G_q), rt(\tau^G_{q^\prime}) \in [t_1-2d,t_2]$ by {\small
[IA-3B]}.

\item $rt(\tau^G_r)\le rt(\tau_r),$ by {\small [IA-3C]}, and if the
inequality $rt(\tau_r)-rt(\tau^G_r)\le \BYZdur$ would not hold, the
node would abort right away.
\end{enumerate}

\item (validity)
If all correct nodes invoked the protocol in an interval
$[t_0,t_0+d],$ as a result of (Initiator$,G,m$) sent by a correct
$G$ that spaced the sending by $6d$ from its last agreement, then
for every correct node $q$ that may have decided $3d$ later than
$G$, the new invocation will still happen with fresh data
structures, since they are reset $3d$ after decision. By that time
it already reset the data structures (including \latest) of the last
execution, and the new decision time $\tau_q$, satisfies $t_0-d \le
rt(\tau^G_q)\le rt(\tau_q)\le t_0+4d$ as implied by {\small
[IA-1D]}.

\item (separation) By {\small [IA-4]} the real-times of the \iacpts satisfy the requirements. Since a node will not 
reset its data
structures before terminating the protocol, it will not send a
$support$ before completing the previous protocol execution.
Therefore, the protocol itself can only increase the time difference
between agreements. Thus, the minimal difference is achieved when a
decision takes place right after the termination of the
primitive \initiator.
\end{enumerate}
%\vspace{-0.5cm}\mbox{\ }\hfill\mbox{\ } 
\end{proof}

%%%%%%%%%%%%%%%%%%%%%%%%%%%%%%%%%%%%%%%%%%%%%%%%
%\newpage
%%%%%%%%%%%%%%%%%%%%%%%%%%%%%%%%%%%%%%%%%%%%%%%%

%\vspace{-0.15cm}
\section*{Acknowledgements}
We wish to thank Ittai Abraham and Ezra Hoch for discussing some of
the fine points of the model and the proofs. This research was
supported in part by grants from ISF, NSF, CCR, and AFOSR.

\vspace{+0.4cm}

\end{document}